\documentclass[11pt,a4paper]{article}
\usepackage{amsmath}
\usepackage{graphicx}
\usepackage{enumerate}
\usepackage{titling}
\usepackage{natbib}
\usepackage{url} 
\usepackage[T1]{fontenc}
\usepackage[utf8]{inputenc}
\usepackage{algorithm}
\usepackage{algpseudocode}
\usepackage{stmaryrd}
\usepackage{subcaption}
\usepackage{amssymb}
\usepackage{bbm}
\usepackage{verbatim}
\usepackage{amsmath, amsfonts, amsthm, amscd}      
\usepackage{fancybox}
\usepackage{tikz}
\usepackage{tikz-qtree}
\usepackage{breqn}
\usepackage{multirow}
\usepackage{mathtools} 
\usepackage{xr-hyper}
\usepackage[hyperindex]{hyperref}
\usetikzlibrary{bayesnet}

\newcommand{\blind}{1}

\newcommand{\R}{\mathbb{R}}

\newcommand{\E}{\mathbb{E}}
\newcommand{\I}{\mathbb{I}}

\newcommand{\var}{\hbox{var}}
\newcommand{\balpha}{{\boldsymbol{\alpha}}}
\newcommand{\sign}{\operatorname{sgn}}

\newtheorem{theorem}{Theorem}[subsection]
\newtheorem{proposition}[theorem]{Proposition}

\newtheorem{lemma}[theorem]{Lemma}
\newtheorem{definition}{Definition}
\newtheorem{remark}{Remark}

\begin{document}

\def\spacingset#1{\renewcommand{\baselinestretch}%
{#1}\small\normalsize} \spacingset{1}


\if1\blind
{
  \title{\bf Robust leave-one-out cross-validation for high-dimensional Bayesian models}
  \author{Luca Alessandro Silva\hspace{.2cm}\\
    Department of Decision Sciences, Bocconi University
\hspace{.2cm}\\    and Giacomo Zanella\thanks{GZ acknowledges support from the European Research Council (ERC), through \textit{StG “PrSc-HDBayLe” grant ID 101076564.}}\hspace{.2cm}\\
    Department of Decision Sciences
    and BIDSA, Bocconi University
    }
\date{}
  \maketitle
} \fi

\if0\blind
{\ 

  \bigskip
  \bigskip
  \bigskip
  \begin{center}
    {\LARGE\bf Robust leave-one-out cross-validation for\\\vspace{5mm} high-dimensional Bayesian models}
\end{center}
  \medskip
} \fi

\bigskip
\begin{abstract}
Leave-one-out cross-validation (LOO-CV) is a popular method for estimating out-of-sample predictive accuracy. 
However, computing LOO-CV criteria can be computationally expensive due to the need to fit the model multiple times. 
In the Bayesian context, importance sampling provides a possible solution but classical approaches can easily produce estimators whose asymptotic variance is infinite, making them potentially unreliable. 
Here we propose and analyze a novel mixture estimator to compute Bayesian LOO-CV criteria.
Our method retains the simplicity and computational convenience of classical approaches, while guaranteeing finite asymptotic variance of the resulting estimators. 
Both theoretical and numerical results are provided to illustrate the improved robustness and efficiency. The computational benefits are particularly significant in high-dimensional problems, allowing to perform Bayesian LOO-CV for a broader range of models, and datasets with highly influential observations. 
The proposed methodology is easily implementable in standard probabilistic programming software and has a computational cost roughly equivalent to fitting the original model once.
\end{abstract}

\noindent%

\spacingset{1.2} 

\section{Introduction}
\label{sec:intro}
Consider a Bayesian model with conditionally independent observations $y=(y_1,\dots,y_n)$ given a set of parameters $\theta$, and denote the resulting joint distribution of $\theta$ and $y$ as
\begin{align}\label{eq:conditionally_ind}
p(\theta,y)=p(\theta)\prod_{i=1}^np(y_i|\theta)\,.
\end{align}
Given some observed data $y$, the model yields a posterior distribution over the unknown parameters, $p(\theta|y)$, and a posterior predictive distribution at a new point $y_{new}$ given by
\begin{align}\label{eq:pred_def}
p(y_{new}|y)&=\int p(y_{new}|\theta)p(\theta|y)d\theta\,.
\end{align}
In various contexts, such as model comparison and selection, one is interested in quantifying the out-of-sample performances of such predictive distributions.
Assuming the existence of a true data-generating process $p^*$, a common measure of predictive performance is the expected log predictive density (ELPD) defined as
\begin{equation}\label{eq:elpd}
\hbox{ELPD}=
\int \log p(y_{new}|y) p^*(y_{new}) dy_{new}.
\end{equation}
Equivalently, one can think at the negative ELPD as the generalization error under a logarithmic loss or (up to additive constants) as the Kullback-Leibler divergence between $p^*$ and the predictive distribution in \eqref{eq:pred_def}. See e.g.\ \citet{WAIC1} and references therein.


The true data generating distribution $p^*$ is unknown in practice and predictive measures such as \eqref{eq:elpd} need to be estimated from observed data. A naive sample average estimate, however, would lead to double use of data. A better way to approximate \eqref{eq:elpd} is cross-validation (CV). In particular, leave-one-out (LOO) CV leads to an estimator 
of ($n$ times) the ELPD defined as
\begin{equation}
\label{eq:lppdloo}
\psi:=
\sum_{i=1}^{n}\log{p(y_i|y_{-i})}
=
\sum_{i=1}^{n}\log\left(\int p(y_i|\theta)p(\theta|y_{-i})d\theta\right)\,,
\end{equation}
where $y_{-i}=(y_{j})_{j\neq i}$, which constitutes a useful criteria to evaluate Bayesian predictive performances \citep{gel1994,vehtari2012survey}. 
See Section \ref{sec:model_eval} for a review on the topic and a discussion of alternative Bayesian model selection criteria.

The main focus of the paper is developing efficient computational methods to approximate $\psi$.
Naive approaches 
require to fit the original model $n$ times, one for each LOO dataset $y_{-i}$, thus being computationally infeasible.
When using Monte Carlo methods to perform computations, a classical solution is to draw samples from $p(\theta|y)$ only once and then resort to importance sampling to approximate each LOO posterior $p(\theta|y_{-i})$ \citep{gel1992}. 
However, as previously noted in the literature, the resulting estimators of $\{p(y_i|y_{-i})\}_{i=1}^n$ are often unreliable and can easily have infinite asymptotic variance \citep{per1997,epi2008}. 
Here we propose novel estimators of $\{p(y_i|y_{-i})\}_{i=1}^n$, based on a mixture representation of leave-one-out posteriors. 
Our method provides guarantees on the finiteness of the estimator's 
asymptotic variance and performs dramatically better than standard competitors in high-dimensional problems, where 
most alternative methodologies break down (see e.g.\ results in Sections \ref{sec:theory} and \ref{sec:sim}). 
Crucially, our methodology requires only a single additional sampling procedure and it can be trivially implemented in probabilistic programming languages, thus preserving the practicality and limited computational cost of previously proposed and widely used solutions \citep{gel1992,vehtari2016}, while offering drastically improved robustness to high-dimensional scenarios. 

More generally, our work supports recent evidence, both in the Bayesian and frequentist literature \citep{beirami2017optimal,rad2020scalable,giordano2019swiss,
paananen2021implicitly}, that LOO-CV criteria can be reliably approximated with a computational cost comparable to the one of a single model fit. 
In this sense LOO-CV can be computationally cheaper than $k$-fold CV by a factor of $k$, since the latter requires fitting $k$ separate models and is not easily amenable to the same importance sampling procedures as LOO-CV.
Such $k$-times speed-up can be crucial in the context of Bayesian computation with Monte Carlo methods where each model fitting can be expensive.
In this context, our work contributes to prevent one of the main factor limiting the applicability of Bayesian LOO-CV, i.e.\ the potential instability of classical estimators of $\{p(y_i|y_{-i})\}_{i=1}^n$.

The article is organized as follows: after briefly reviewing Bayesian model selection criteria in Section \ref{sec:model_eval}, we describe our proposed computational methodology and compare it to classical ones in Section \ref{sec:methodology}. Section \ref{sec:theory} provides some theoretical analysis of the resulting estimators, including a proof of finite asymptotic variance and a comparison to classical methods in high-dimensional regression contexts.
 Section \ref{sec:sim} provides numerical results that support the theoretical findings and illustrate the improved robustness both to the presence of model misspecification and to high-dimensionality of the parameter space.
Finally, Section \ref{sec:ext} discusses potential extensions of our methodology (e.g.\ different scoring rules or non conditionally-independent models). 
For notational brevity, throughout the paper we use the same letter $p$ to denote appropriate joint, marginal and conditional distributions of the model for $\theta$, $y$ and $y_{new}$, as done in \eqref{eq:conditionally_ind}-\eqref{eq:lppdloo}.
Similarly, we leave the dependence of $p(y_i|\theta)$ and $p(y_{new}|\theta)$ on additional covariates or other variables implicit in the notation.

\subsection{Predictive criteria for Bayesian model comparison}\label{sec:model_eval}
Bayesian model selection criteria are often divided into ones based on posterior model probabilities, such as Bayes factors and classical Bayesian model averaging \citep{hoeting1999bayesian}, and ones based more directly on predictive distributions \citep{box1980,gel1994,WAIC1,vehtari2012survey}.
In line with the classical tension between identification and estimation in model selection \citep{yang2005}, the two approaches have complementary roles, with the first class of methods being more naturally suited to model identification and the second to maximise predictive accuracy. 
Appealing features of predictive-based criteria include being more directly comparable across different models (including non-nested ones), and being typically less sensitive to prior specifications compared to Bayes factors, including vague priors as in e.g.\ Bartlett's paradox in Bayesian model selection
 \citealp{bartlett1957comment,lindley1957statistical,liang2008mixtures}. 
 The literature on the topic is vast and we refer to \citet{gel1994,vehtari2012survey} and references therein for an overview and some arguments in favour of Bayesian predictive measures and cross-validation criteria. 
 
While there exist various scoring functions to evaluate predictive distributions \citep{gne2011}, in this paper we focus on the logarithmic one as in \eqref{eq:elpd}, which is the unique 
 local and proper scoring rule \citep{ber1979} and the most widely used in practice. 
 See for example \cite{gelman2014understanding} for arguments in favour of using the ELPD metric in \eqref{eq:elpd} and its LOO-CV estimator in \eqref{eq:lppdloo}. 
 Beyond computing ELPD estimates as in \eqref{eq:lppdloo}, LOO predictive probabilities $\{p(y_i|y_{-i})\}_{i=1}^n$ are also of interest in themselves, as they allow to implement methodologies aimed at optimizing predictive performances such as Bayesian stacking \citep{yao2018using} 
or at identifying discording observations \citep{pettit1990conditional,weiss1998bayesian} to guide model improvements and refinements.

A direct alternative to CV is the use of information criteria, which can also be thought at as approximations to generalization losses or out-of-sample prediction measures \citep{Stone1977}. 
 In particular, the Widely Applicable Information Criteria (WAIC) has been shown to be asymptotically equivalent (as $n\to\infty$) 
to $\psi$ under weak assumptions \citep{WAIC1}. 
See also \citet{vehtari2016} for a comparison of WAIC and $\psi$. 
Note that, while classical information criteria based on point estimates \citep{AIC,schwarz1978} tend to be computationally much cheaper than CV, more elaborate Bayesian criteria such as the Deviance Information Criteria (DIC) \citep{DIC} and WAIC also require Monte Carlo samples from the posterior, thus being closer to Bayesian LOO-CV in terms of computational cost.

Other types of CV schemes, such as $k$-fold, are also often used instead of LOO and the best choice in terms of statistical properties is in general case-dependent. 
LOO tends to have smaller bias compared to $k$-fold with $k\ll n$, while the ordering among their variances is less obvious and more case-dependent, see \citet[Sec.5]{arlot2010} and references therein. 
While most results on statistical properties of CV estimators are in non-Bayesian settings, \citet{watanabe2009,WAIC1} provide bounds on the difference in expectation between $\psi/n$ and ELPD, as $n\to\infty$ with fixed dimensionality.  
More recently \citet[Coroll.4.3]{patil2021} prove consistency (uniformly w.r.t.\ hyper-parameters) of LOO-CV estimators of prediction error for high-dimensional Ridge regression. 
Their results suggest good statistical properties of $\psi$ for high-dimensional Bayesian linear regression models, being thus closer to the high-dimensional settings of Sections \ref{sec:theory} and \ref{sec:sim}.
Finally note that, while this work focuses on LOO-CV, Section \ref{sec:ext} discusses extensions to leave-$p$-out CV with $p>1$.

\section{Computing Bayesian leave-one-out cross validation}
\label{sec:methodology}
In this paper we focus on Monte Carlo methodologies to compute the LOO predictive probabilities $\{p(y_i|y_{-i})\}_{i=1}^n$. 
Depending on the context, these may be themselves the quantities of interest, or an intermediate step to compute LOO-CV criteria such as 
$
\psi
$
 defined in \eqref{eq:lppdloo}.
 In the latter case an estimate of $\psi$ is simply obtained by plugging-in the estimates of $p(y_i|y_{-i})$ in \eqref{eq:lppdloo}.

The first, somehow brute-force, approach to this computation would be to fit $n$ times the model separately.
Recalling that $p(y_i|y_{-i})=\int p(y_i|\theta)p(\theta|y_{-i})d\theta $, one could draw $S$ Monte Carlo samples from each LOO posterior $p(\theta|y_{-i})$, using e.g.\ $n$ separate MCMC runs, and then estimate $p(y_i|y_{-i})$ with the resulting sample average of $p(y_i|\theta)$. 
We denote the resulting estimators of $\mu_i:=p(y_i|y_{-i})$ as
\begin{equation}\label{eq:loo_est}
\hat{\mu}_i^{(loo)}
=
S^{-1}\sum_{s=1}^{S}p(y_i|\theta_s)\,,
\end{equation}
where $\theta_1, \theta_2,...,\theta_S$ are samples from $p(\theta | y_{-i})$.
Assuming the computational cost of each Monte Carlo sample to grow linearly with $n$, 
this would require $\Theta(Sn)$ samples and  $\Theta(Sn^2)$ computational cost in total, which is typically unfeasible.

A potential solution proposed in \citep{gel1992} is to instead draw only one set of samples from the full-data posterior, and then use importance sampling to approximate expectations with respect to the $n$ different LOO posteriors.
 This leads to unnormalized importance weights between the $i$-th LOO posterior and the full posterior equal to
\begin{equation*}
w_i^{(post)}(\theta)=p(y_i|\theta)^{-1}\propto \frac{p(\theta|y_{-i})}{p(\theta|y)}\,.
\end{equation*}
The corresponding self-normalized importance sampling estimator of $p(y_i|y_{-i})$ is
\begin{equation}\label{eq:post_est}
\hat{\mu}_i^{(post)}
=
\frac{
\sum_{s=1}^{S}p(y_i|\theta_s)w_i^{(post)}(\theta_s)
}{\sum_{s=1}^{S}w_i^{(post)}(\theta_s)
}
=
\left(S^{-1}\sum_{s=1}^{S}p(y_i|\theta_s)^{-1}\right)^{-1},
\end{equation}
where $\theta_1, \theta_2,...,\theta_S$ are samples from $p(\theta | y)$.
This procedure is practically appealing because it only requires one sampling routine and has $\Theta(Sn)$ total cost, including the computation of the $n$ estimators $\{\hat{\mu}_i^{(post)}\}_{i=1}^n$, each of which can be obtained at $\Theta(S)$ cost given the samples $\{\theta_s\}_{s=1}^S$ using definition \eqref{eq:post_est}.
The drawback is that the resulting importance sampling estimators can be unstable and even have infinite variance.
In such cases the estimators are still consistent, i.e.\ $\lim_{S\to\infty}\hat{\mu}_i^{(post)}=p(y_i|y_{-i})$ almost surely, but the 
asymptotic normality and the $S^{-1/2}$ rate of convergence may not hold \citep{epi2008}.
These issues are not surprising if one realizes that \eqref{eq:post_est} is a variation of the classical harmonic-mean estimator \citep{newton1994approximate}, which has well-known stability issues. 
This has motivated proposals in the literature to improve the stability of LOO-CV estimators as well as to diagnose their potential failure. A notable example that we compare with in simulations later on is the Pareto-smoothed importance sampling (PSIS) methodology of \citep{vehtari2016} implemented in the popular \emph{loo} R package \citep{loo_package}. 
See also \citet{alqallaf2001cross,bornn2010efficient,rischard2018unbiased,paananen2021implicitly} for other work in the area, and Section \ref{sec:alternatives} for comparison with some of those. 

\subsection{Mixture estimators}\label{sec:mix_prop}
Here we propose a different set of estimators for $\{p(y_i|y_{-i})\}_{i=1}^n$ with drastically improved robustness to high-dimensionality, which we achieve by expressing the problem in terms of mixtures rather than harmonic mean identities.
We introduce a component indicator $I$, formally a random variable on $\{1,\dots,n\}$, and define a joint distribution for $\theta$ and $I$ as 
\begin{align}\label{eq:joint_q_mix}
q_{mix}(\theta,I)=&\frac{p(\theta)p(y_{-I}|\theta)}{
\sum_{j=1}^np(y_{-j})}
&(\theta,I)\in\Theta\times \{1,\dots,n\}\,.
\end{align}
Here $q_{mix}$ is defined so that $q_{mix}(\theta|I=i)=p(\theta|y_{-i})$ and thus $p(y_{i}|y_{-i})$ can be written as the following conditional expectation
\begin{align}
 p(y_{i}|y_{-i})=\E_{(\theta,I)\sim q_{mix}}[p(y_i|\theta)|I=i]\,.
\end{align}
This representation 
 leads to our proposed set of estimators, which are obtained through the following steps:
\begin{itemize}
\item[(i)] draw $S$ samples $\theta_1,\theta_2,...,\theta_S$ from $ q_{mix}(\theta)$, where 
 \begin{align}
q_{mix}(\theta)=
Z^{-1}
\sum_{j=1}^n p(\theta)p(y_{-j}|\theta)  \propto p(\theta|y) \left(\sum_{j=1}^n p(y_{j}|\theta)^{-1}\right)\,, \label{eq:mixture}
\end{align}
is the marginal distribution of $\theta$ under the joint $q_{mix}(\theta,I)$ and $Z=\sum_{j=1}^n p(y_{-j})$. Sampling from \eqref{eq:mixture} can be done using standard MCMC algorithms, as discussed below;
\item[(ii)] for each $i\in\{1,\dots,n\}$, obtain weighted samples from $p(\theta|y_{-i})$ assigning to each sample in $\{\theta_s\}_{s=1,\dots,S}$  the weight
\begin{equation*}
w_i^{(mix)}(\theta)=q_{mix}(I=i|\theta)=\frac{p(y_i|\theta)^{-1}}{\sum_{j=1}^np(y_j|\theta)^{-1}}\,,
\end{equation*}
which is  the conditional probability of $I=i$ given $\theta$ under the joint distribution $q_{mix}(\theta,I)$;
\item[(iii)] for each $i\in\{1,\dots,n\}$, estimate $p(y_i|y_{-i})$ with
\begin{equation}\label{eq:mix_est}
\hat{\mu}_i^{(mix)}
=
\frac{
\sum_{s=1}^{S}p(y_i|\theta_s)w_i^{(mix)}(\theta_s)
}{\sum_{s=1}^{S}w_i^{(mix)}(\theta_s)
}\,.
\end{equation}
\end{itemize}

The estimator in \eqref{eq:mix_est} can also be interpreted as a self-normalized importance sampling estimator with importance distribution $q_{mix}(\theta)$ and target distribution $p(\theta|y_{-i})$, so that $w_i^{(mix)}(\theta)$ are unnormalized importance weights between target and importance distribution.
We often use this formulation when proving theoretical results in Section \ref{sec:theory}.

The proposed estimators $\{\hat{\mu}_i^{(mix)}\}_{i=1}^n$ retain the simplicity and computational practicality of the classical ones in \eqref{eq:post_est}. In fact a single sampling routine is required, this time from $q_{mix}(\theta)$, and the total computational cost to obtain the $n$ estimators $\{\hat{\mu}_i^{(mix)}\}_{i=1}^n$ is still $\Theta(Sn)$.
The latter follows from two crucial remarks. First, evaluating $q_{mix}(\theta)$ up to normalizing constant requires $\Theta(n)$ cost using the last expression in \eqref{eq:mixture}, see also \eqref{eq:log_mixture} in the Supplement. Note that a naive use of the first expression in \eqref{eq:mixture} would instead incur in a $\Theta(n^2)$ cost. Second, computing $\{\hat{\mu}_i^{(mix)}\}_{i=1}^n$ in \eqref{eq:mix_est} requires first an $\Theta(Sn)$ computation common to all $i$'s, namely the computation of 
$\{\sum_{j=1}^np(y_j|\theta_s)^{-1}\}_{s=1}^S$ and, given the latter, each $\hat{\mu}_i^{(mix)}$ can be computed at $\Theta(S)$ cost. 
See Section \ref{sec:eff_comp_estim} in the Supplement for more details.

Also, evaluating gradients of the log of the mixture distribution, $\nabla \log q_{mix}(\theta)$, involves an $\Theta(n)$ cost, analogously to gradients of the standard log-posterior $\nabla \log p(\theta|y)$, and the whole routine is trivial to implement in probabilistic programming languages that rely on gradient-based MCMC, such as \textsc{stan} \citep{stan}.
In our numerical experiments, sampling from $p(\theta|y)$ and $q_{mix}$ with \textsc{stan} required a comparable amount of time, with only a slight overhead for $q_{mix}$. 
See Section \ref{sec:sampling_mix_implem} in the Supplement for more details on efficient and numerically stable implementation of the sampling procedure.
A practical advantage of the estimators $\{\hat{\mu}_i^{(post)}\}_{i=1}^n$, however, is that they can re-use samples from $p(\theta|y)$ that might have already been drawn to perform posterior inferences.  
On the contrary, computing $\{\hat{\mu}_i^{(mix)}\}_{i=1}^n$ requires a separate sampling routine, specific for ELPD estimation purposes. 


\begin{remark}[Mixture interpretation]\label{rmk:mix}
The distribution $q_{mix}$ can be interpreted as a mixture of LOO posteriors writing $q_{mix}(\theta)=\sum_{i=1}^n\pi_i p(\theta|y_{-i})$,
with mixture probabilities 
$
\pi_i
=
Z^{-1}p(y_{-i})
$ satisfying $\sum_i \pi_i=1$ and $\pi_i\geq 0$.
Rewriting $\pi_i=\tilde{Z}^{-1}p(y_i|y_{-i})^{-1}$, with $\tilde{Z}=\sum_jp(y_j|y_{-j})^{-1}$, we can express the quantity of interest as $p(y_i|y_{-i})=\tilde{Z}^{-1}/\pi_i$.
Indeed, the denominator in \eqref{eq:mix_est} times $S^{-1}$ is a consistent estimator of $\pi_i$ while the numerator times $S^{-1}$ is a consistent estimator of $\tilde{Z}^{-1}$. 
 Thus, the algorithm is effectively estimating the probability $\pi_i$ of each component in the mixture representation and using that to estimate $p(y_i|y_{-i})$.
This is arguably where the improvement in performances of $\hat{\mu}_i^{(mix)}$ compared to $\hat{\mu}_i^{(post)}$ comes from, since mixture probabilities are typically easier to estimate than harmonic means. 
For example the weights $w_i^{(mix)}(\theta)$ are by construction upper bounded by $1$, being conditional probabilities, which is a desirable feature to improve robustness of importance sampling estimators.
\end{remark}

The idea of using mixtures to derive estimators with improved stability underlies various methodologies in the Monte Carlo literature, such as Bridge Sampling and variations \citep{bennett1976efficient,geyer1991estimating,meng1996simulating,shirts2008statistically}. In this sense, one can think at our proposed methodology as an effective and practical way to extend these techniques to LOO-CV computation contexts while preserving a $\Theta(Sn)$ total computational cost.

\begin{remark}[Choice of mixture probabilities]\label{rmk:weights}
Note that the mixture probabilities $\pi_i$ involve the intractable terms $p(y_{-i})$ that are typically not available in closed form.
However, these terms cancel with the denominator of $p(\theta|y_{-i})=p(\theta)p(y_{-i}|\theta)/p(y_{-i})$, making $q_{mix}(\theta)$ computable up to a single intractable normalizing constant $Z$ as in \eqref{eq:mixture} and thus amenable to standard sampling algorithms. 
Also, since $\pi_i\propto p(y_i|y_{-i})^{-1}$, $q_{mix}$ naturally puts more weight on mixture components with smaller $p(y_i|y_{-i})$. 
This is desirable since small values of $p(y_i|y_{-i})$ are typically harder to estimate and contribute more to $\psi=\sum_{i=1}^{n}\log{p(y_i|y_{-i})}$. 
In Sections \ref{sec:fin_var} and \ref{sec:ext} we discuss extensions to mixture constructions with general choices of mixture probabilities.
\end{remark}

\begin{remark}[Potential multimodality of $q_{mix}$]\label{rmk:multimod}
As mentioned above, sampling from $q_{mix}(\theta)$ or from $p(\theta|y)$ usually requires comparable computational effort.
The main circumstance where this is not true is when $q_{mix}$ features significant multimodality while $p(\theta|y)$ does not. This could happen because $q_{mix}$ is a mixture. 
However, unlike usual mixtures considered in statistical contexts, $q_{mix}$ is an average of a large number of similar distributions (i.e.\ the $n$ LOO posteriors) rather than a small number of well-separated ones.
In most situations, this averaging makes the marginal distribution $q_{mix}$ smooth and well behaved.
\end{remark}


\section{Analysis of the proposed estimator}\label{sec:theory}
 In this section we provide a theoretical analysis of the proposed estimators $\{\hat{\mu}_i^{(mix)}\}_{i=1}^n$ with particular focus on comparing them with the classical ones $\{\hat{\mu}_i^{(post)}\}_{i=1}^n$.
We measure the efficiency of estimators by their relative asymptotic variances, defined as 
\begin{equation}
\label{eq:as_var_def}
 AV^{(post)}_i=\lim_{S\to\infty}S\,\var (\hat{\mu}_i^{(post)}/\mu_i)
\quad\hbox{and}\quad
 AV^{(mix)}_i=\lim_{S\to\infty}S\,\var (\hat{\mu}_i^{(mix)}/\mu_i)\,,
\end{equation}
where $\mu_i=p(y_i|y_{-i})$ as before.
By the delta method we also have 
$$AV^{(post)}_i=\lim_{S\to\infty}S\,\var\big(\log (\hat{\mu}_i^{(post)})\big) \quad\hbox{and}\quad AV^{(mix)}_i=\lim_{S\to\infty}S\,\var\big(\log (\hat{\mu}_i^{(mix)})\big),$$ meaning that the above terms can also be interpreted as the asymptotic variances of the plug-in estimators on the log-scale, $\log (\hat{\mu}_i^{(post)})$ and $\log (\hat{\mu}_i^{(mix)})$.
Thus $\{AV^{(post)}_i\}_{i=1}^n$ and $\{AV^{(mix)}_i\}_{i=1}^n$ are a natural measure of performance when the quantities of interest are $\{\log(p(y_i|y_{-i}))\}_{i=1}^n$ or $\psi$ in \eqref{eq:lppdloo}. 

Note that the asymptotic variances in \eqref{eq:as_var_def} refer to the case when $(\theta_s)_{s=1}^{S}$ in \eqref{eq:post_est} and \eqref{eq:mix_est} are i.i.d.\ samples from, respectively, $p(\theta|y)$ and $q_{mix}(\theta)$.
In practice, one is rarely able to draw i.i.d.\ samples from such distributions and instead typically relies on MCMC schemes, leading to correlated samples. 
In such cases the asymptotic variances of the actual estimators used in practice can be decomposed as the product of an importance sampling contribution times an MCMC contribution, namely as the product of the asymptotic variances in \eqref{eq:as_var_def} times an MCMC integrated autocorrelation time, see e.g.\ Lemma 1 of \citet{zanella2019scalable}.  
Thus, while formally referring to the i.i.d.\ case, the asymptotic variances in \eqref{eq:as_var_def} are 
 relevant also to the case of MCMC sampling.

\subsection{Finiteness of asymptotic variances}\label{sec:fin_var}
As mentioned above, a serious drawback of the classical estimator is that its
asymptotic variance $AV^{(post)}_i$ can be very large, even infinite.
Indeed \citet{per1997,epi2008} provide various examples, even simple ones, where $AV^{(post)}_i$ is infinite. Our first key theoretical result states that, on the contrary, the proposed mixture estimators lead to finite asymptotic variances under minimal technical assumptions. 
In particular, we will only require that
\begin{equation}\label{eq:reg_assumptions}
\tag{A1}
p(y_i|y_{-i})>0\hbox{ and }\int_{\Theta}p(y_i|\theta)p(\theta|y)d\theta<\infty\hbox{ for all }i=1,\dots,n\,.
\end{equation}
The above assumptions require the quantity of interest $p(y_i|y_{-i})$ to be non-zero, otherwise $\log p(y_i|y_{-i})$ would not be well defined, and the predictive distribution $p(y_{new}|y)$ based on the full data to be finite at  $y_{new}=y_i$ for each $i$.
These are minimal assumptions that are typically satisfied for any model where LOO-CV quantities are of interest.
Given these, we can state the following result.
\begin{theorem}
\label{thm:mix_finite}
Under \eqref{eq:reg_assumptions} we have that $AV_i^{(mix)}<\infty$ for all $i=1,\dots,n$, and $MSE_i^{(mix)}:=\E[(\log(\hat{\mu}_i^{(mix)})-\log(\mu_i))^2]$ satisfies
\begin{align}\label{eq:MSE_log_scale}
MSE_i^{(mix)}
&=
AV_i^{(mix)}\times S^{-1}+\mathcal{O}(S^{-2})&\hbox{as }S\to\infty\,.
\end{align}
\end{theorem}
\color{black}
Theorem \ref{thm:mix_finite} holds also in the more general case where $q_{mix}$ in \eqref{eq:mixture} is replaced by a weighted version 
$q_{mix}^{(\balpha)}(\theta)=Z_{\balpha}^{-1}\sum_{i=1}^{n}\alpha_ip(y_{-i}|\theta)p(\theta)$ where $Z_{\balpha}=\sum_{i=1}^{n}\alpha_ip(y_{-i})$ and $\balpha=(\alpha_i)_{i=1}^n$ are arbitrary weights satisfying $\alpha_i\in(0,\infty)$ for all $i$. 
In the supplement we prove the result in such more general version.
See also Remark \ref{rmk:weights} and Section \ref{sec:ext} for more details on the practical relevance of the more general weighted mixture $q_{mix}^{(\balpha)}$.

Theorem \ref{thm:mix_finite} highlights a first sharp distinction between the classical and mixture estimators.
It shows that, under minimal assumptions,  $\hat{\mu}_i^{(mix)}$ enjoys finite asymptotic variance and thus its mean squared error (MSE) decays at the usual $\mathcal{O}(S^{-1})$ rate as $S\to\infty$. 
The latter follows from the fact that 
the squared bias of $\hat{\mu}_i^{(mix)}$ is of order $\mathcal{O}(S^{-2})$ and thus variance dominates the MSE for large $S$. 
On the contrary the MSE of classical estimators $\hat{\mu}_i^{(post)}$ may decay at a slower than $\mathcal{O}(S^{-1})$ rate, which is indeed observed in practice even for simple models, see e.g.\ Figure \ref{evolution:mse} in Section \ref{sec:dep_on_S}. In such situations, the improvement in efficiency between $\hat{\mu}_i^{(mix)}$ and $\hat{\mu}_i^{(post)}$ increases to infinity as $S\to\infty$.

Note that $AV^{(mix)}_i<\infty$ is a stronger requirement than $\var (\hat{\mu}_i^{(mix)})<\infty$ for fixed $S$, since the latter does not give guarantees on rate of decay with $S$. 
For example, methods based on truncation or smoothing of the importance weights, such as PSIS estimators mentioned above, 
lead to estimators with finite variance for fixed $S$ but whose MSE can decay at a slower than $\mathcal{O}(S^{-1})$ rate, see e.g.\ \citet{vehtari2022pareto} for more details.

\subsection{High-dimensional regression models}\label{sec:linear}
In this section we provide a more refined analysis of the behavior of $AV^{(post)}_i$  and $AV^{(mix)}_i$, focusing on high-dimensional regression models, first considering the linear case and then a more general regression context. 
Our results suggest that classical estimators are highly sensitive to high-dimensionality and their performances quickly deteriorate as the ratio $p/n$ increases, while the mixture estimator exhibits substantially improved robustness.

\subsubsection{Connection to Bayesian leverage and the impact of high-dimensionality}
Consider the regression model 
\begin{equation}
\label{modelcov}
\begin{gathered}
y_i|\theta \sim N(x_i^T\theta,\sigma^2)\qquad i=1,\dots,n\\
\theta\sim N(\theta_0,\Sigma)\,,\qquad\qquad\qquad
\end{gathered}
\end{equation}
where $x_i$ and $\theta$ indicate $p\times 1$ matrices of, respectively, covariates and parameters.
We assume the noise level $\sigma^2$ and the prior mean and covariance, $\theta_0$ and $\Sigma$, to be fixed and known.
For the linear model in \eqref{modelcov}, the finiteness of  $AV^{(post)}_i$ is elegantly related to the notion of Bayesian leverage.
Denoting by $X$ the $n\times p$ matrix of covariates, define the \emph{Bayesian hat matrix}, or \emph{Ridge hat matrix}, as
\begin{equation}\label{eq:hat_mat}
H=X(X^TX+\sigma^2\Sigma^{-1})^{-1}X^T\,,
\end{equation}
which collapses to the standard (frequentist) hat matrix in the flat prior case, i.e.\ when $\Sigma^{-1}=0$.
The diagonal element $H_{ii}$ represents the Bayesian leverage of the $i$-th observation.
Thus, a higher value of $H_{ii}$ indicates higher discrepancy between the full and LOO posteriors, $p(\theta|y)$ and $p(\theta|y_{-i})$, which in turn implies that the importance sampling estimator in \eqref{eq:post_est} can have poor behavior. 
The theorem below makes the connection precise. 
The connection between leverages and the finiteness of $AV^{(post)}_i$ was previously studied in \citet{per1997}.
The following result extends results therein, allowing for $p>n$ and using the notion of Bayesian leverage, rather than the frequentist one (which corresponds to $\Sigma^{-1}=0$).
\begin{theorem}
\label{thm:leverage}
Under \eqref{modelcov}, for each $i\in\{1,\dots,n\}$, we have $AV^{(post)}_i<\infty$ if and only if $H_{ii}<0.5$.
\end{theorem}
The connection to Bayesian leverages provides useful insight in the behavior of the classical estimator in \eqref{eq:post_est} and in particular on its dependence with respect to the dimensionality of $\theta$ and the amount of prior shrinkage. 
Consider first the case of flat improper prior for $\theta$, corresponding to $p<n$ and $\Sigma^{-1}= 0$. In such case $H$ is the standard (frequentist) hat matrix and its trace satisfies $\sum_{i=1}^nH_{ii}=rk(X)$, where $rk(X)$ denotes the rank of $X$.
For linearly independent predictors we have $rk(X)=p$, which implies that $H_{ii}\geq p/n$ for at least one $i$. Thus, by Theorem \ref{thm:leverage}, as soon as $p\geq n/2$ some $AV^{(post)}_i$ must be infinite. 
When the entries of $X$ are random variables (r.v.s) with complex 
Gaussian distributions, it holds $H_{ii}\sim Beta(p,n-p)$, see Appendix A of \citet{chave2003bounded}. 
This provides a more refined description of leverages distribution under a random design assumption and further highlights the key role of the ratio $p/n$, since there $E[H_{ii}]=p/n$.
The same holds by symmetry for any random design that is exchangeable over rows of $X$ and gives $rk(X)=p$ almost surely.
This is consistent with our numerical experiments, where the performances of classical estimators quickly degrade as $p$ increases and degenerate when $p$ is of the same order as $n$.

More generally, when $\Sigma=\nu^2\I_p$, with $\I_p$ being the $p\times p$ identity matrix, 
each $H_{ii}$ is a strictly decreasing function of the so-called ridge regularization parameter $\lambda=\sigma^2\nu^{-2}$ and the trace of $H$ satisfies $\sum_{i=1}^nH_{ii}=\sum_{j=1}^{rk(X)}\frac{d_j^2}{d_j^2+\lambda}$, where $(d_j)_{j=1}^{rk(X)}$ are the singular values of $X$ \citep{walker1988influence}.
Thus, increasing the amount of prior regularization lowers the values of the Bayesian leverages, increasing the chances of having $AV^{(post)}_i<\infty$ for all $i$.
This is consistent with the intuition that stronger shrinkage and regularization decreases the sensitivity of the posterior to each single observation, making LOO-CV calculations potentially easier. 
Nonetheless, as illustrated in Figure \ref{fig:leverages} we see below, even under strong prior shrinkage the leverages $H_{ii}$ can be large when $p/n$ is large, leading to instability of classical estimators. 
 \begin{figure}[hbt]
\centering
\begin{tikzpicture}
 \node (img1) {\includegraphics[width=0.32\textwidth]{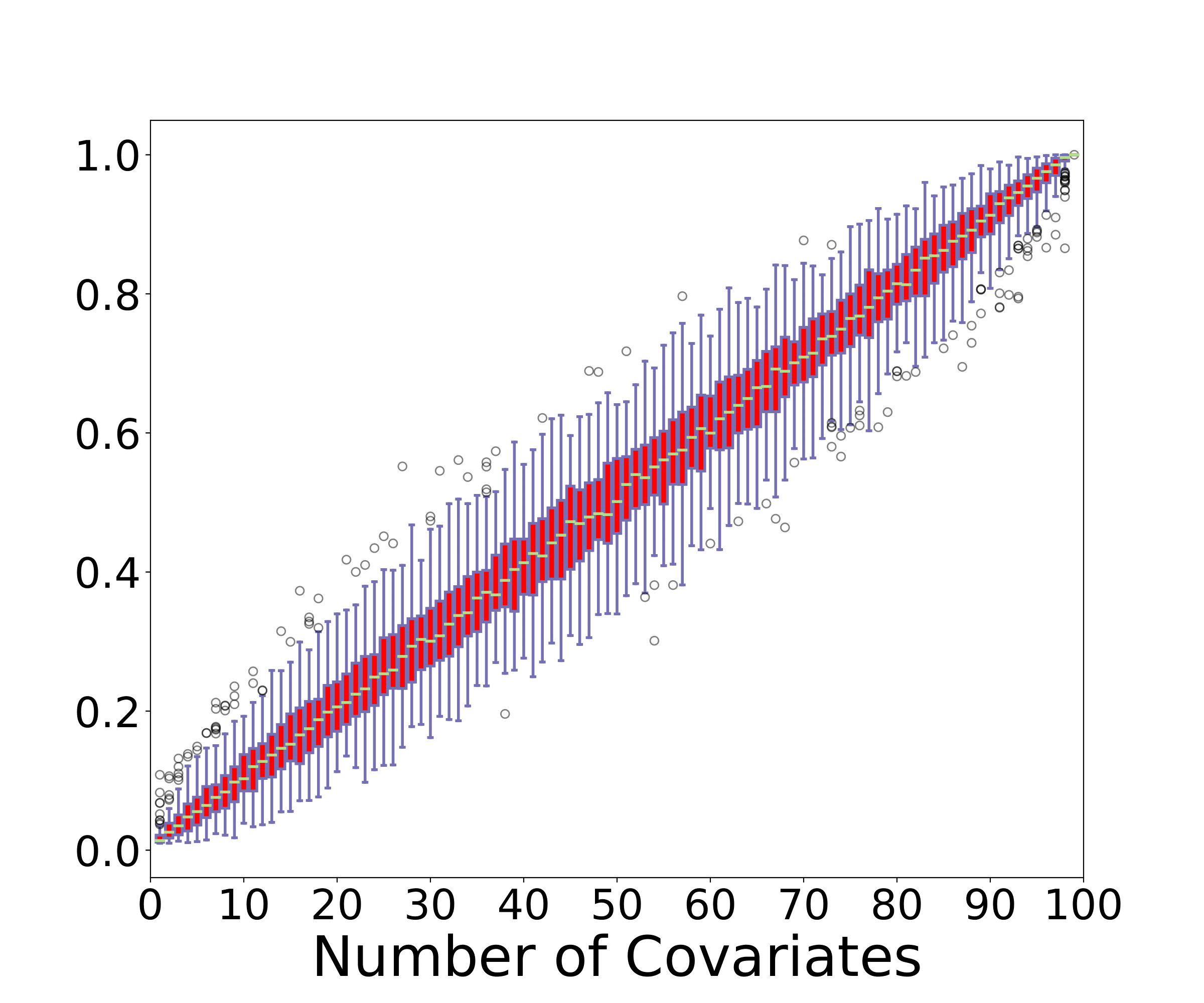}};
 \node[left=of img1, node distance=0cm, rotate=90, anchor=center,yshift=-1.1cm,font=\color{black}] {\small{Leverages}};
 \node[right=of img1, yshift=-0.0cm, xshift=-1.0cm](img2) {\includegraphics[width=0.32\textwidth]{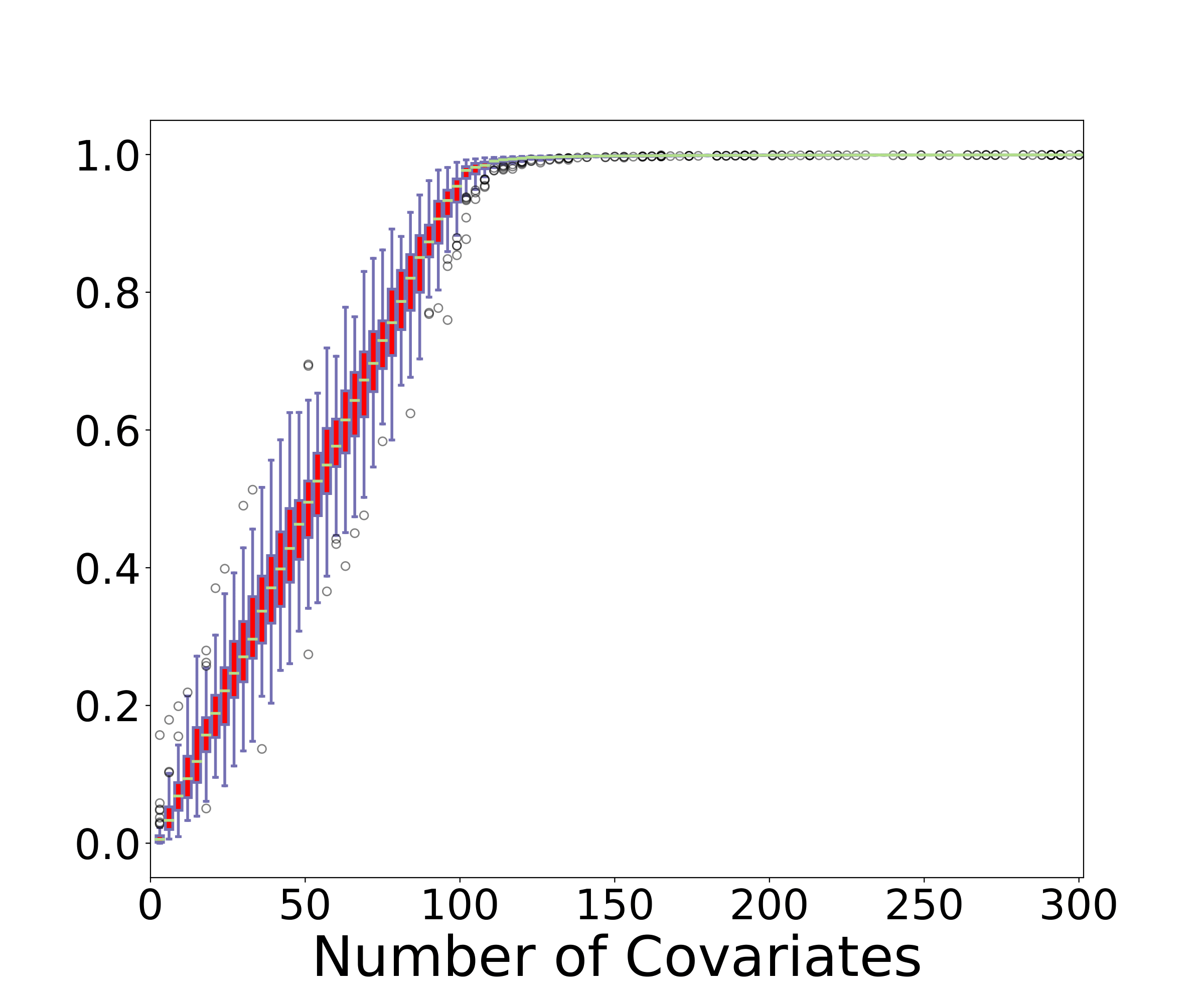}};
  \node[left=of img2, node distance=0cm, rotate=90, anchor=center,yshift=-1.1cm,font=\color{black}] {\small{Leverages}};
 \node[right=of img2, yshift=0.0cm, xshift=-1.0cm](img3) {\includegraphics[width=0.32\textwidth]{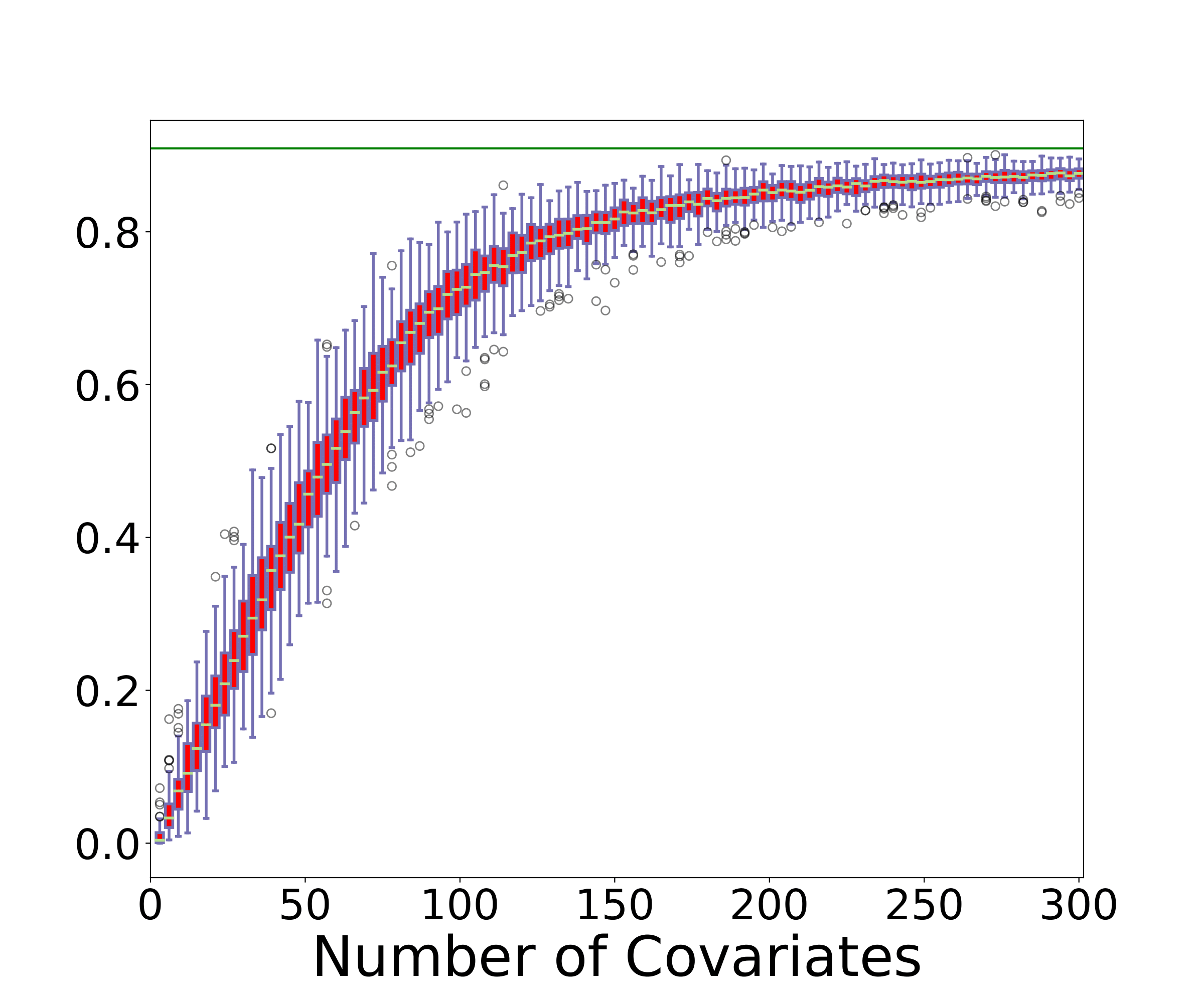}};
  \node[left=of img3, node distance=0cm, rotate=90, anchor=center,yshift=-1.1cm,font=\color{black}] {\small{Leverages}};
\end{tikzpicture}
 \caption{Distribution of the leverages $\{H_{ii}\}_{i=1,\dots,n}$ as a function of $p$ for $n=100$ and $(X_{ij})_{i,j}\stackrel{iid}\sim N(0,1)$. Left: $\Sigma^{-1}=0$, center: $\Sigma=10\cdot \I_p$, right: $\Sigma=10/p\cdot \I_p$.}
 \label{fig:leverages}
 \end{figure}

\subsubsection{Behavior of the classical and mixture estimators in large $p$ regimes}\label{sec:high_d_linear}
We now provide a high-dimensional asymptotic analysis of $AV^{(post)}_i$ and $AV^{(mix)}_i$ under random design assumptions. Specifically, we assume that
\begin{equation}
(X_{ij})_{i,j\geq 1}\hbox{ are independent 
 r.v.s with }E[X_{ij}]=0,\;Var(X_{ij})=\tau^2<\infty\hbox{ and }E[X_{ij}^4]\leq c_x
\label{eq:rand_des}
\tag{A2}
\end{equation}
for some $c_x<\infty$.
The assumption of zero mean and constant variance is realistic in settings where the regressors are standardized. Finiteness of fourth moments is used to derive appropriate strong law of large numbers for $XX^T$ without assuming identically distributed covariates. While the assumption of independence is potentially restrictive, it allows to derive more intuitive and explicit results. 
We expect our conclusions to hold well beyond such assumption but we leave such extensions (e.g.\ to cases of weak dependence among predictors, such as Assumption 3 in \citealp{fasano2019scalable}) to future work.

We consider settings where $p$ can be large. 
In such cases, it may be appropriate to assume the prior covariance of $\theta$ to vary with $p$.
An interesting and natural setting is to take $\Sigma=\nu_p^2\I_p$ with $\nu_p^2=c/p$ for some fixed $c>0$, which induces a prior variance of the linear predictors $\var(x_i^T\theta)=c(p^{-1}\sum_{j=1}^pX_{ij}^2)$ that is approximately constant w.r.t.\ $p$ and converges to the non-degenerate value $c\tau^2\in(0,\infty)$ as $p\to\infty$ under \eqref{eq:rand_des}. 
Other regimes considered in the literature are ones where $\nu_p^2$ is constant or where it scales as $\Theta(n/p)$. 
The following proposition characterizes the behaviour of $H_{ii}$ when $p\to\infty$ for all such cases, which can be obtained with different choices of $c$.
\begin{proposition}
\label{prop:high_p}
Assume \eqref{modelcov} and \eqref{eq:rand_des}, with $\Sigma=\nu_p^2\I_p$ and $\lim_{p\to\infty} p\nu_p^2=c\in[0,\infty]$. For each $i\in\{1,\dots,n\}$, we have
\begin{align}\label{eq:high_p}
H_{ii}&\;\to\;
\frac{c\tau^2}{\sigma^2+c\tau^2}\qquad \hbox{almost surely as }p\to\infty.
\end{align}
In the above convergence $n$ is fixed while $p\to\infty$, and $\frac{c\tau^2}{\sigma^2+c\tau^2}=1$ when $c=\infty$.
It follows that $AV^{(post)}_i=\infty$ almost surely for large enough $p$ if $c\tau^2>\sigma^2$, while $\limsup_{p\to\infty}AV^{(post)}_i<\infty$ almost surely if $c\tau^2<\sigma^2$.
\end{proposition}

The statement about $AV^{(post)}_i$ being eventually infinite for a large enough $p$ when $c\tau^2> \sigma^2$ is a direct consequence of \eqref{eq:high_p} and Theorem \ref{thm:leverage}. 
This is coherent with the numerical simulations of Section \ref{sec:sim}, where the classical estimator eventually breaks down as $p/n$ increases.
The condition $c\tau^2> \sigma^2$ is satisfied for most common prior specifications.
It is obviously satisfied when $\nu$ is constant since $c=\infty$ there.
Under stronger prior shrinkage where $\nu_p^2=c/p$ with $c<\infty$, one typically sets $c$ to some value that is significantly larger that the noise variance $\sigma^2$, to avoid overly informative priors for the linear predictors $x_i^T\theta$, and thus $c\tau^2> \sigma^2$ will typically hold also there. 
Finally, the condition $c\tau^2> \sigma^2$ can be directly interpreted as a comparison between prior and likelihood information, in particular as requiring the latter to be stronger than the former. 

Taking the limit for $p\to\infty$ when $n$ is fixed mimics a regime where $p$ is large compared to $n$. 
As shown in the simulations of Section \ref{sec:sim}, such regime is highly challenging for Monte Carlo methods performing LOO-CV computations, the intuition being that the discrepancy among LOO posteriors is maximal in such regime. 

We now study the behaviour of $AV^{(mix)}_i$ in settings similar to Proposition \ref{prop:high_p}.
We first consider the case where $c<\infty$.
\begin{theorem}
\label{thm:limsup_fin}
Assume \eqref{modelcov} and \eqref{eq:rand_des}, with $\Sigma=\nu_p^2\I_p$ and $\lim_{p\to\infty} p\nu_p^2=c\in[0,\infty)$. Then we have $\limsup_{p\to\infty}AV^{(mix)}_i<\infty$ almost surely for every $i\in\{1,\dots,n\}$.
\end{theorem} 
Compared to Theorem \ref{thm:mix_finite}, which guarantees that $AV^{(mix)}_i<\infty$ for every fixed dataset  and thus for every $p$, Theorem \ref{thm:limsup_fin} proves the stronger statement that each $AV^{(mix)}_i$ is also uniformly bounded with respect to $p$, suggesting that mixture estimators are remarkably robust to high-dimensionality of the parameter space.

\subsubsection{More general regression models}\label{sec:high_d_glm}
We now extend some of the results derived above for the Gaussian model \eqref{modelcov} to more general regression contexts. The results suggest that the improved robustness of $\hat{\mu}_i^{(mix)}$ compared to $\hat{\mu}_i^{(post)}$, especially in high-dimensions, is not specific to Gaussian likelihoods but rather it holds more generally. We consider regression models with general likelihood and Gaussian prior, where
\begin{equation}
\begin{gathered}
\theta\sim N(\theta_0,\Sigma)\,,\qquad\qquad\qquad
\\
p(y|\theta)
=
\prod_{i=1}^n
g(y_i|\eta_i)\,,
\qquad
\hbox{ where }\eta_i=x_i^T\theta 
\hbox{ for }i=1,\dots,n\,,
\end{gathered}\label{eq:glm_model}
\end{equation}
and $g(\cdot|\cdot):\R\times\R\to[0,\infty)$ is a generic likelihood function. The above formulation includes generalized linear models (GLM's) with Gaussian prior.
Throughout, we assume the likelihood to be upper bounded, i.e.\ $\sup_{\eta_i}g(y_i|\eta_i)<\infty$ for any fixed $y_i\in\R$. The latter is arguably a mild assumption that is typically satisfied in practice.
\begin{theorem}
\label{thm:limsup_fin_general}
Assume \eqref{eq:glm_model} and \eqref{eq:rand_des}, with $\Sigma=\nu_p^2\I_p$ and $\lim_{p\to\infty} p\nu_p^2=c\in[0,\infty)$. Then we have that almost surely, for each $i\in\{1,\dots,n\}$:\\
(a) $\limsup_{p\to\infty}AV^{(mix)}_i<\infty$\\
(b) $\limsup_{p\to\infty}AV^{(post)}_i<\infty$ if 
\begin{equation}\label{eq:light_lik}
\int\exp\left(-\delta \eta_i^2
\right)g(y_i|\eta_i)^{-1}d\eta_i<\infty\,,
\end{equation}
for some $\delta<(2c\tau^2)^{-1}$, while $AV^{(post)}_i=\infty$ for large enough $p$ if the integral in \eqref{eq:light_lik} is equal to infinity for some $\delta>(2c\tau^2)^{-1}$.
\end{theorem} 

Theorem \ref{thm:limsup_fin_general} extends the results of Section \ref{sec:high_d_linear} to generic likelihoods. 
Namely $AV^{(mix)}_i$ is shown to remain bounded away from infinity as $p$ grows, while $AV^{(post)}_i$ is shown to become eventually equal to $\infty$ when \eqref{eq:light_lik} does not hold, i.e.\ provided the likelihood function has light enough tails.
 In the Gaussian likelihood case, \eqref{eq:light_lik} coincides with requiring $\sigma^2>c\tau^2$, which directly relates to Proposition \ref{prop:high_p} and the discussion thereafter. 
Condition \eqref{eq:light_lik} also relates to the study of $AV^{(post)}_i$ under thick-tail or light-tail priors in \citet{epi2008}, although there the opposite scenario is considered where the likelihood is Gaussian and the prior is general and no asymptotic regime is considered.

Finally, we consider the case where the prior variance of the linear predictors diverges with $p$, i.e.\ $\lim_{p\to\infty} p\nu_p^2=\infty$.
This happens for example when $\nu_p^2$ remains constant as $p\to\infty$.
In this case $AV_{i}^{(mix)}$ can also diverge as $p\to\infty$, depending on the tail behavior of the likelihood function. 
The underlying reason is that in such cases the LOO predictive probabilities $p(y_i|y_{-i})$ go to $0$ as $p\to\infty$ and even the asymptotic variance of the LOO estimators $\hat{\mu}_i^{(loo)}$, which we regard as the gold-standard but computationally expensive approach, diverge.
We denote $AV_{i}^{(loo)}=\lim_{S\to\infty}S\,\var (\hat{\mu}_i^{(loo)}/\mu_i)$ in the next theorem.
\begin{theorem}
\label{thm:av_loo_general}
Assume \eqref{eq:glm_model} and \eqref{eq:rand_des}, with $\Sigma=\nu_p^2\I_p$ and $\lim_{p\to\infty} p\nu_p^2=\infty$. Then: 
\\
(a) 
if 
$\int g(y_i|\eta_i)d\eta_i<\infty$ for $i=1,\dots,n$
then 
 $\lim_{p\to\infty}AV_{i}^{(loo)}=\lim_{p\to\infty}AV_{i}^{(mix)}=\infty$ almost surely for $i=1,\dots,n$;\\
(b) if
\begin{equation}\label{eq:pos_lim}
\lim_{\eta_i\to\infty}g(y_i|\eta_i)+\lim_{\eta_i\to -\infty}g(y_i|\eta_i)\in(0,\infty)\qquad \hbox{for }i=1,\dots,n 
\end{equation}
then $\limsup_{p\to\infty}AV^{(mix)}_i<\infty$ and  $\limsup_{p\to\infty}AV^{(loo)}_i<\infty\;\;$ almost surely as $p\to\infty$ for $i=1,\dots,n$. If \eqref{eq:pos_lim} holds and $\lim_{\eta_i\to\infty}g(y_i|\eta_i)=0$ or $\lim_{\eta_i\to -\infty}g(y_i|\eta_i)=0$ for $i=1,\dots,n$,
then $\lim_{p\to\infty}AV_i^{(post)}=\infty$.
\end{theorem}

Theorem \ref{thm:av_loo_general} shows that, when $\lim_{p\to\infty} p\nu_p^2=\infty$, the asymptotic behaviour of $AV_{i}^{(mix)}$, as well as $AV_{i}^{(loo)}$, depends on the type of likelihood in the model.
For integrable likelihoods, i.e.\ ones satisfying $\int g(y_i|\eta_i)d\eta_i<\infty$ such as for Gaussian, Poisson, etc.,  the performances of all estimators under consideration (including the mixture and the gold-standard but expensive LOO ones) deteriorate as $p\to\infty$, see case (a) of Theorem \ref{thm:av_loo_general}.
As mentioned above, the deterioration of performances of the mixture and LOO estimators in this case is related to the target probabilities $p(y_i|y_{-i})$ going to $0$ as $p\to\infty$.
Instead, for non-integrable likelihoods such as the logistic one, which falls into case (b) of Theorem \ref{thm:av_loo_general}, we have that $\lim_{p\to\infty}AV_i^{(post)}=\infty$ while $\limsup_{p\to\infty}AV^{(mix)}_i<\infty$.

\section{Numerical simulations and real data examples}\label{sec:sim}
In this section we provide extensive numerical simulations, both on synthetic and real data, to compare the efficiency of the classical and mixture estimators. 
We also include 
PSIS estimators \citep{vehtari2016} in the comparison, which is the default methodology implemented in the Bayesian LOO-CV  R package \emph{loo} \citep{loo_package}. 

We test the above estimators in challenging cases where difficulty in computing LOO predictive probabilities arises mostly from two sources: (a) high-dimensionality of the parameter space and (b) model misspecification and presence of observations that are not well fit by the model.
We test (a) by considering large $p$ scenarios and (b) by considering real datasets with either known or potential observations not well fit by the model. 
The results suggest that mixture estimators dominates classical and PSIS ones and, in line with the theoretical results of Section \ref{sec:theory}, that the magnitude of the improvement increases with the dimensionality of problem, while also being potentially large for low dimensional problems with highly influential observations (see e.g.\ the examples in Section \ref{sec:leuk_stack} of the Supplement). 

In addition, Section \ref{sec:alternatives} provides comparisons to the methodologies in \citet{alqallaf2001cross} and \citet{bornn2010efficient}, while Sections \ref{sec:bias_var_dec} and \ref{sec:leuk_stack} in the Supplement provide, respectively, numerical illustrations of the bias-variance decomposition for the MSE of the estimators under consideration and tests on 
the \emph{Leukaemia} and \emph{Stack Loss} datasets, which are standard examples in the literature on Bayesian LOO-CV computation \citep{per1997,epi2008,vehtari2016,rischard2018unbiased}.

\subsection{High-dimensional linear regression}
We start by considering high-dimensional linear regression models, where the quantities of interest $\{p(y_i|y_{-i})\}_{i=1}^n$ can be computed in closed form and the different estimators can be compared in terms of the induced mean squared errors (MSE) for a variety of setting. 

\subsubsection{Dependence of the estimators efficiency on $n$ and $p$}\label{sec:high_d_linregr_sims}

First we explore how the performances of the different estimators depend on the number of data points $n$ and parameters $p$. 
We consider the model in \eqref{modelcov}, with $\sigma^2=1$ and two prior specifications, one where $\Sigma=\I_p$ and one where $\Sigma=100/p\cdot\I_p$. We take $n\in\{50, 100, 150\}$ and for every such value we vary $p/n$ ranging from $0.1$ to $3$. For every resulting $(n,p)$ pair we generate $10^4$ synthetic datasets, simulating the design matrix $X$ with i.i.d.\ standard normal entries (plus an intercept) and the data $y$ from the corresponding model likelihood in \eqref{modelcov}.
For each generated dataset, we compute the exact values of $\{p(y_i|y_{-i})\}_{i=1}^n$, as well as the corresponding classical, mixture and PSIS estimators based on $S=2\times 10^3$ i.i.d.\ samples from either $p(\theta|y)$ or $q_{mix}(\theta)$.
We compute the PSIS estimator using the \textsc{python} code available at \url{https://github.com/avehtari/PSIS}. 
We then compute the MSE of the estimators on the log scale, e.g.\ 
$\E[(\log(\hat{\mu}_i^{(post)})-\log(\mu_i))^2]$
 for the classical estimator. For each $(n,p)$ pair we report the MSE averaging both over datasets and over $i=1,\dots,n$. 

The results are reported in Figure \ref{fig:mse}.
PSIS estimators mildly improve over the classical ones for small-to-moderate ratios $p/n$ but overall the two perform similarly. 
\begin{figure}[t!]
  \includegraphics[width=.5\linewidth]{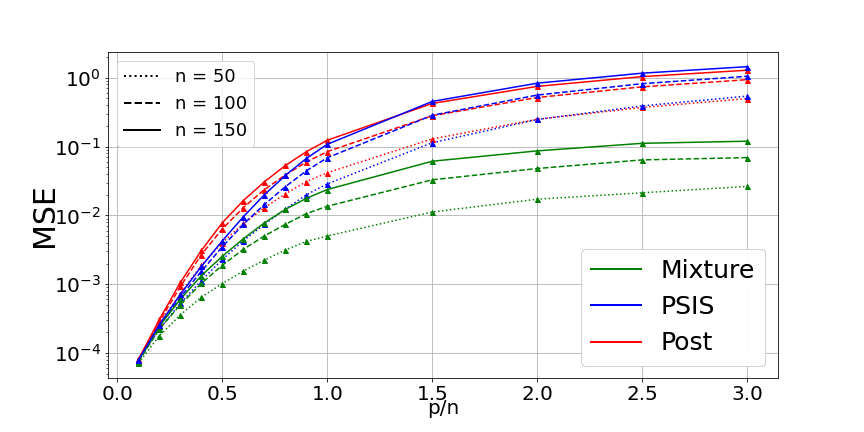}\hfill
  \includegraphics[width=.5\linewidth]{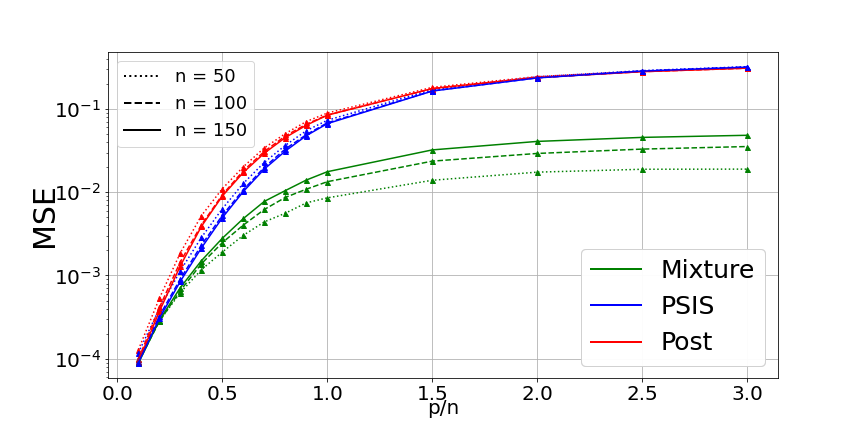}
\caption{MSE for posterior, PSIS and mixture estimators of $\{\log p(y_i|y_{-i})\}_{i=1}^n$ for high-dimensional linear regression models with different values of $n$, $p$ and prior variance (left: $\Sigma=\I_p$; right: $\Sigma=100/p\cdot\I_p$). 
 See Section \ref{sec:high_d_linregr_sims} for more details.}
   \label{fig:mse}
 \end{figure}
For example, the MSE of PSIS is never smaller than the one of posterior by more than a factor of $2$, with largest reduction in MSE being roughly of $40\%$ for values of $p/n\approx 0.35$.
Mixture estimators outperform both posterior and PSIS ones, with improvements that increase with the ratio $p/n$. 
In such high-dimensional regimes classical and PSIS estimators break down (note the log-scale) while mixture estimators remains reliable with moderate MSE. This is in agreement with the theory in Section \ref{sec:theory}, which shows that $AV_i^{(post)}$ becomes infinite for $p$ sufficiently large, while $AV_i^{(mix)}$ is finite and uniformly bounded with respect to $p$ when $\Sigma=c/p\I_p$ with $c>0$. 
All methods perform better when the prior is more informative, i.e.\ when $\Sigma=100/p\cdot\I_p$ compared to $\Sigma=\I_p$, which is again in accordance with 
Section \ref{sec:theory}.

\subsubsection{Infinite asymptotic variance and failure of standard rate of convergence}\label{sec:dep_on_S}
Next we explore more directly the impact of having a finite versus infinite asymptotic variance.
Since the latter corresponds to a slower than $\mathcal{O}(S^{-1})$ decay for the MSE (see Theorem 
\ref{thm:mix_finite} and discussion thereafter), the difference is better illustrated by fixing $n$ and $p$ and varying $S$. 
\begin{figure}[h!]
\center
\includegraphics[width=.7\linewidth]{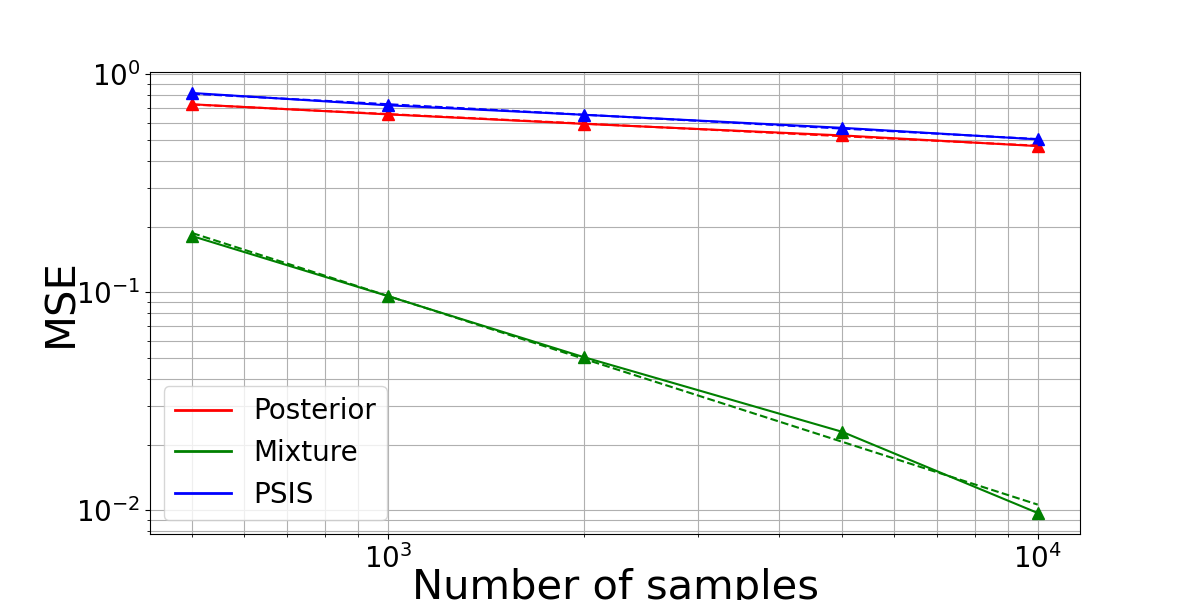}
\caption{MSE decay (solid lines) w.r.t.\ number of samples $S$. 
Dashed lines represent linear fits and have slopes of -0.957 for mixture, -0.145 for posterior and -0.160 for PSIS. 
See Section \ref{sec:dep_on_S} for more details.
}
\label{evolution:mse}
\end{figure}
We thus consider the same set-up and MSE computation of Section \ref{sec:high_d_linregr_sims}, but now we vary $S$ while fixing $p=n=100$ and $\Sigma=\I_p$. Figure \ref{evolution:mse} reports the results. 
Section \ref{sec:theory} implies that in this setting $AV_i^{(post)}=\infty$ with high probability while $AV_i^{(mix)}<\infty$.
In accordance with this, we observe an MSE of the classical and PSIS estimators decaying approximately at a rate  $\mathcal{O}(S^{-0.1})$ and an MSE of the mixture estimators following the theoretical $\mathcal{O}(S^{-1})$ rate.
In practice, this means that in such scenarios, despite being consistent as $S\to\infty$, classical and PSIS estimators will require an extremely large number of samples to make the MSE small. 

\subsubsection{Real data, misspecification and non-conjugate priors}\label{sec:bladder}
We now move to study how our estimator performs in a regression setting on a real dataset. We consider the $\textit{Bladder cancer}$ data available in the Gene Expression Omnibus (GEO) repository at $\textit{https://www.ncbi.nlm.nih.gov/gds}$, with accession number $GSE31684$. The full dataset has $93$ observations, and for every observation, we have $54680$ covariates, most of which are gene expressions of the patients. 
We derive different sub-datasets with varying $p/n$ ratios by taking the first $p$ covariates of the original dataset for $p\in\{\frac{n}{2}, n, 2n, 3n, 4n, 5n\}$.
For each of the resulting six datasets, we standardize covariates and response variable to have zero mean and unit variance before fitting the model. 
First, we employ the usual Bayesian linear regression model with conjugate prior, 
$y| X, \sigma^2 \sim N(X\theta, \sigma^2\I_n)$ and $\theta|\sigma^2\sim N(\theta_0, \sigma^2\Sigma)$,
with $\theta_0=0$, $\Sigma=100/p\I_p$ and set $\sigma^2 = \arg\!\max_\sigma p(y|\sigma^2)$ in an empirical Bayes fashion. 
The latter operation was not needed for synthetic data, where $\sigma$ was set to the true data-generating value.
Note that the value of $\sigma^2$ influences the prior variance for $\theta$ as indicated in the above model specification.

We compute estimators based on $S=2\times 10^4$ i.i.d.\ samples from either $p(\theta|y)$ or $q_{mix}$. 
Table \ref{table:bladder} reports the resulting MSE, both average and maximum w.r.t.\ 
$i=1,\dots,n$, averaged over $100$ independent repetitions. 
Here the MSE values are significantly larger than the ones for simulated data with similar dimensionality and data size (compare e.g.\ Figure \ref{fig:mse} with Figure \ref{fig:bladder} in the Supplement), suggesting that real data and model misspecification make LOO-CV computations harder.
\begingroup
\setlength{\tabcolsep}{10pt} 
\renewcommand{\arraystretch}{0.6}
\begin{table}[t!]
\begin{center}
\begin{tabular}{| p{2.5cm} |p{1.7cm} | p{2.5cm} | p{2.5cm} | p{1.7cm} |}
\hline
num. of cov. & Estimator & $n^{-1}\sum\limits_{i=1}^nMSE_i$ & $\max\limits_{i\in\{1,..,n\}}MSE_i$ & $\%$ $k$$>.7$\\
\hline
\multirow{3}{3cm}{$p=n/2$} & Mixture & 1.1e-03 & 7.0e-03 & -\\
\cline{2-2}
& Posterior & 1.5e-01 & 5.4e-01 & 24$\%$\\
\cline{2-2}
& PSIS & 1.7e-01 & 2.4e-01 & -\\ 
\hline
\multirow{3}{3cm}{$p=n$} & Mixture & 2.9e-01 & 1.5e+00& -\\
\cline{2-2}
& Posterior & 2.8e+00 & 6.1e+00 & 86$\%$\\
\cline{2-2}
& PSIS & 3.1e+00 & 4.1e+00 & \\ 
\hline
\multirow{3}{3cm}{$p=2n$} & Mixture & 7.9e-02 & 4.0e-01 & -\\
\cline{2-2}
& Posterior & 2.6e+00 & 5.9e+01 & 99$\%$\\
\cline{2-2}
& PSIS & 2.9e+00 & 3.7e+00 & -\\ 
\hline
\multirow{3}{3cm}{$p=5n$} & Mixture & 2.9e-02 & 1.2e-01 & -\\
\cline{2-2}
& Posterior & 2.1e+00 & 4.9e+00 & 99$\%$\\
\cline{2-2}
& PSIS & 2.4e+00 & 3.0e+00 & -\\
\hline
\end{tabular}
\end{center}
\caption{MSE for subsets of the Bladder dataset with increasing dimensionality under a conjugate linear regression model. 
MSE$_i$ refers to $\E[(\log (\hat\mu_i)-\log (\mu_i))^2]$ where $\hat\mu_i$ is the estimator of $\mu_i$ under consideration, while $k>.7$ refers to the diagnostic produced by the \emph{loo} R package \citep{loo_package}.
See Section \ref{sec:bladder} for more details.
}\label{table:bladder}
\end{table}
\endgroup
Table \ref{table:bladder} reports also the 
percentages of data points with large Pareto shape parameter $k$ computed with the \emph{loo} R package \citep{loo_package} which is commonly used to diagnose instability of the estimators $\hat\mu_i^{(post)}$.

Finally, we consider non-conjugate priors, namely independent Laplace, or double-Exponential, priors for $\theta_1,\dots,\theta_p$ with mean parameter equal to $0 $ and scale parameter equal to $b=\sqrt{50/p}$, so to have prior variance for each coefficient equal to $100/p$. We keep a Gaussian likelihood, $ y| X, \sigma^2 \sim N(X\theta, \sigma^2\I_n)$, treating the noise parameter $\sigma$ as unknown and assigning a $InvGamma(4,6)$ prior to it. 
We consider the subset of the Bladder data with $p=2n$. 
Non-conjugate high-dimensional problems are challenging for Bayesian LOO-CV computations based on importance sampling and indeed most examples considered in the literature are of low or moderate dimensionality, with exceptions including \citep{lamnisos2012cross,paananen2021implicitly}.
Since the model is not conjugate the true values are not available and thus we computed an accurate approximation to those that we use as benchmark, using leave-one-out estimators based on long MCMC runs (namely using $10$ chains with $8\times 10^3$ samples each, resulting in $4\times 10^4$ total samples after discarding the first half as `burn-in' or `warm-up').
 To ensure high quality of the samples both from the posterior and the mixture we set the \textsc{stan} control values to $adapt\_delta=0.99$ and $max\_treedepth=15$ respectively.
\begingroup
\setlength{\tabcolsep}{10pt} 
\renewcommand{\arraystretch}{0.6}
\begin{table}[t!]
\begin{center}
\begin{tabular}{| p{2.5cm} |p{1.7cm} | p{2.5cm} | p{2.5cm} | p{1.7cm} |}
\hline
num. of cov. & Estimator & $n^{-1}\sum\limits_{i=1}^nMSE_i$ & $\max\limits_{i\in\{1,..,n\}}MSE_i$ & $\%$ $k$$>.7$\\
\hline
\multirow{3}{3cm}{$p=2n$\\(Laplace prior)} & Mixture & 3.0e-0.2 & 2.7e-0.1 & -\\
\cline{2-2}
& Posterior & 5.6e-01 & 3.6e+00 & 86$\%$\\
\cline{2-2}
& PSIS & 6.1e-01 & 2.9e+00 & -\\ 
\hline
\end{tabular}
\end{center}
\caption{
MSE in estimating $\{\log (p(y_i|y_{-i}))\}_{i=1}^n$ for a linear regression model with non-conjugate Laplace prior on the Bladder dataset. See Section \ref{sec:bladder} for more details.
}\label{table:bladder_laplace}
\end{table}
\endgroup
We then compute 25 independent replications of the posterior and mixture estimators based on the default \textsc{stan} value of $S=4\times 10^3$ and report the resulting MSE in Table \ref{table:bladder_laplace}. 
In this example mixture estimators provide roughly a 20 times reduction in MSE compared to the posterior ones. 

\subsection{High-dimensional binary regression}\label{sec:high_dim_bin}
We now consider three high-dimensional binary regression examples. 
We consider three real datasets from the UCI machine learning repository at \url{https://archive.ics.uci.edu/}, namely the \emph{Arrhythmia}, \emph{Voice} and \emph{Parkinson} ones, 
which cover different $n/p$ 
 ratios. 
Preprocessing of the data included 
 removal of covariates that were almost equal for all individuals, which created stability problems to the HMC algorithm implemented in \textsc{stan} especially for the \emph{Arrhythmia} dataset, and normalisation of all remaining covariates to have zero mean and unit variance.
 The values of $(n,p)$ for the three datasets in their final format, which can be found at \url{https://github.com/luchinoprince/Mixture_IS}, are $(452,208)$ for Arrythmia, $(756,755)$ for Parkinson and $(126,312)$ for Voice.

For each dataset we ran four MCMC chains for $2\times 10^3$ iterations each, removing the first half as burn-in, leaving us with $S=4\times 10^3$ samples from both the posterior and the mixture distributions, which were used to compute the classical, mixture and PSIS estimators. \textsc{stan} with defaults setting was used and no convergence or mixing issues were detected with standard diagnostics. 
For the Arrythmia and Voice datasets we obtained accurate estimates (which we treat as ground truth values) for $\{\log (p(y_i|y_{-i}))\}_{i=1}^n$ by drawing $5\times 10^4$ samples from each of the $n$ LOO posteriors separately as done in Section \ref{sec:leuk_stack} of the Supplement.
For the Parkinson dataset, the above procedure would have been computationally unfeasible and we instead obtained ground truth values for $\{\log (p(y_i|y_{-i}))\}_{i=1}^n$ 
running a long chain sampling from $q_{mix}$ and then computing the mixture estimators based on $10^6$ samples. Standard diagnostics suggested that the Monte Carlo error for these estimates was at least one order of magnitude smaller than the one of the other estimates under consideration.

Table \ref{table:high_d} summarizes the resulting MSE of the estimators relative to the ground truth values, averaging over 10 independent repetitions for each combination of dataset and method.
\begingroup
\setlength{\tabcolsep}{10pt} 
\renewcommand{\arraystretch}{0.6}
\begin{table}
\begin{center}
\begin{tabular}{| p{2.4cm} |p{1.6cm} | p{2.3cm} | p{2.3cm} | p{1.8cm} |}
\hline
Dataset  & Estimator & $n^{-1}\sum\limits_{i=1}^nMSE_i$ & $\max\limits_{i\in\{1,..,n\}}MSE_i$  & $\%$ $k$$>.7$\\
\hline
\multirow{3}{3.0cm}{Arrythmia \\ $n$=452, $p$=208}
 & Mixture & 4.4e-03 & 3.9e-01 & -\\
\cline{2-2}
& Bronze & 8.0e-03 & 1.2e+00 & 23$\%$ \\
\cline{2-2}
& Posterior & 9.3e-03 & 1.1e+00 & 25$\%$\\
\cline{2-2}
& PSIS & 6.4e-03 & 8.7e-01 & -\\
\hline
\multirow{3}{3.0cm}{Parkinson\\$n$=756, $p$=755} & Mixture & 3.6e-03 & 3.3e-01 & -\\
\cline{2-2}
& Bronze & 8.7e-03 & 1.2e+00 & $49\%$ \\
\cline{2-2}
& Posterior & 1.0e-02 & 2.0e+00 & 53$\%$\\
\cline{2-2}
& PSIS & 6.0e-03 & 5.0e-01 & -\\
\hline
\multirow{3}{3.0cm}{Voice\\$n$=126, $p$=312} & Mixture & 2.3e-03 & 6.6e-02 & -\\
\cline{2-2}
& Bronze & 2.4e-02 & 1.1e+00 & 54$\%$ \\
\cline{2-2}
& Posterior & 1.8e-02 & 9.7e-01 & 42$\%$\\
\cline{2-2}
& PSIS & 2.0e-02 & 1.0e+00 & -\\
\hline
\end{tabular}
\end{center}
\caption{
{MSE in estimating $\{\log (p(y_i|y_{-i}))\}_{i=1}^n$ for a high-dimensional binary regression model with Laplace prior on three real datasets. 
See Section \ref{sec:high_dim_bin} for more details.}
}\label{table:high_d}
\end{table}
\endgroup
The mixture estimator performs significantly better than both the classical and PSIS estimators in these examples, see also Figure \ref{fig:traces} in the supplement for traceplots of the classical and mixture estimators.
See below for discussion on the bronze estimator also reported in Table \ref{table:high_d}.

\subsubsection{Comparison to additional alternative computational methodologies}\label{sec:alternatives}
In this section we provide a brief comparison with other alternative methodologies from the Bayesian LOO-CV computation literature, using the three datasets of Table \ref{table:high_d}.
We consider the gold, silver and bronze estimators proposed in \citep{alqallaf2001cross} and the
Sequential Monte Carlo (SMC) approach of \citep{bornn2010efficient}.

The $\textit{bronze}$ estimator of \citep{alqallaf2001cross} is the easiest to compare with.
In our framework, such method estimates $\{p(y_i|y_{-i})\}_{i=1}^n$ performing self-normalized importance sampling with importance distribution given by the following tempered posterior
\begin{equation}
\label{bronze}
q_{br}(\theta) \propto \left(\prod_{i=1}^n p(y_i|\theta)\right)^{\frac{n-1}{n}} p(\theta).
\end{equation}
This procedure has a computational cost comparable to the posterior, mixture and PSIS ones for the same number of samples.
We thus test it on the examples in Table \ref{table:high_d} using the same number of samples as well as \textsc{stan} settings.
The resulting MSE, which are reported in Table \ref{table:high_d}, are closer to the ones of the posterior and PSIS estimators rather than the mixture ones. See also Section \ref{sec:temp_est} in the Supplement for more discussion of the bronze estimator and more generally estimators based on geometric tempering.

The SMC methodology of \citep{bornn2010efficient}, when applied to our context, coincides with running $n$ SMC routines, one for each target value $p(y_i|y_{-i})$, initialized from the same samples drawn from the posterior $p(\theta|y)$. 
When an adaptive SMC approach is employed, this procedure ends up performing pure importance sampling (with the posterior as importance distribution) for data points inducing well behaved importance weights (e.g.\ ones with ESS above a given threshold) while performing a genuine SMC routine involving resampling and MCMC moves for the other values.
While the resulting estimators are often guaranteed to have finite variance (see \citealp{bornn2010efficient}), the total computational cost can be quadratic in $n$ if a considerable proportion of data points requires non-trivial SMC routines.
We thus test how many data points require non-trivial SMC routines for the high-dimensional binary regression examples above.
The results suggest that approximately $40\%$ for the \emph{Voice} Dataset, $22\%$ for the \emph{Parkinson} dataset and $64\%$ for the \emph{Arrhythmia} dataset. Such percentages where calculated by looking at the effective sample size (ESS) of the weights of the posterior, and assessing how many where under the threshold of $1/2$, which is a default value commonly used in the literature \citep{chopin2020introduction}. 
Such high percentages suggest that SMC, at least in the above version, is not particularly suited to such a $\textit{cross-sectional}$ estimation procedure, since running $\Theta(n)$ separate SMC routines makes it computationally too demanding, while it can be very appealing in $\textit{longitudinal}$ scenarios, such as hyper-parameter tuning, see e.g.\ \citet{bornn2010efficient}.

Finally, we consider the $\textit{gold}$ and $\textit{silver}$ of \citep{alqallaf2001cross}. 
These allow to obtain only an estimation of the whole LOO-CV sum $\psi=\sum_{i=1}^{n}\log(p(y_i|y_{-i}))$ in \eqref{eq:lppdloo}, as opposed to the $n$ terms $\{p(y_i|y_{-i})\}_{i=1}^n$.
In particular, the $\textit{gold}$ estimator of $\psi$ is defined as
\begin{equation}
\label{gold:est}
\hat{\psi}_{gold} = \frac{n}{K}\displaystyle\sum_{i\in I} \log (p(y_i|y_{-i})), 
\end{equation}
where $K$ is a fixed  integer in $\{1,\dots,n\}$ and $I$ is a collection of $K$ indices uniformly sampled without replacement from $\{1,2,....n\}$. 
The gold estimator is not computable in practice since we do not know the exact values of $p(y_i|y_{-i})$.
A practical approach is given by the so-called silver estimator, which is defined as
\begin{equation}
\label{silv:est}
\hat{\psi}_{silv} = \frac{n}{K}\displaystyle\sum_{i\in I} \log(\hat{\mu}_i^{(loo)}),
\end{equation}
with $K$ and $I$ defined as for the gold estimator and $\hat{\mu}_i^{(loo)}$ as in \eqref{eq:loo_est}.
We compare the silver estimator with the estimator of $\psi$ obtained from the mixture estimators by plug-in, i.e.\ 
$\hat{\psi}_{mix} =\sum_{i=1}^n \log(\hat{\mu}_i^{(mix)})$.
To ensure comparability, we we fix the total computational resources to $2\times 10^4$ samples (including burn-in ones) both for the silver and mixture estimators. Thus, for a given value of $K$, each chain used to compute a single $\hat{\mu}_i^{(loo)}$ has a total of $2\times 10^4/K$ samples.
Figures \ref{silver:all} shows the errors in estimating $\psi$ obtained with $\hat{\psi}_{silv} $ for different values of $K$ and with $\hat{\psi}_{mix}$. We can see that, for small values of $K$, $\hat{\psi}_{silv} $ has a large variance due to the variability in the choice of the subset $I$. On the contrary, as $K$ increases the bias of each estimator $\hat{\mu}_i^{(loo)}$ increases, since these are self-normalized importance sampling estimators based on $2\times 10^4/K$ samples, which became too few samples as $K$ increases (in the extreme case of $K= 721$ for the Parkinson data one has $2\times 10^4/K\approx 28$ samples for every estimator).
As a result, regardless of the value of $K$, $\hat{\psi}_{silv}$ has a much larger estimation error (note the log-scale on the $y$ axis) than $\hat{\psi}_{mix}$ with the same number of total samples.
Note that for the $\textit{Voice}$ dataset, given the small values of $n$ and the large number of total samples, the performances of the silver estimator are monotonically increasing with $K$ and the optimal value is $K=n$, which makes the silver estimator coincide with the brute force approach discussed in Section \ref{sec:methodology}.
\begin{figure}
\includegraphics[scale=0.27]{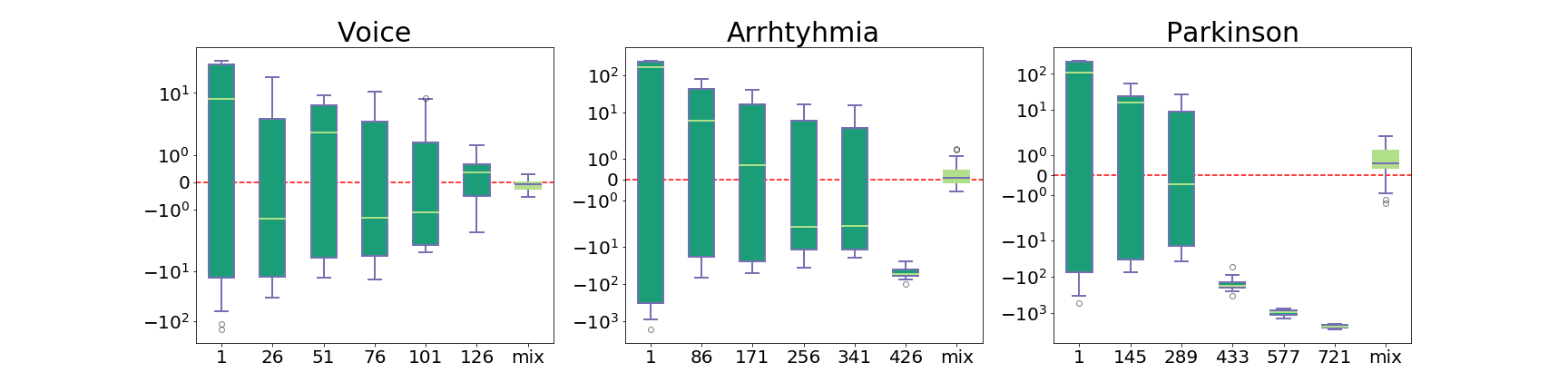}
\caption{Errors in estimating $\psi$ for $\hat{\psi}_{silv}$, with different values of $K$ on the $x$-axis and $\hat{\psi}_{mix}$. Boxplots are based on 25 independent repetitions for each estimator.}
\label{silver:all}
\end{figure}

\section{Extensions}\label{sec:ext}
The proposed mixture estimator can be extended in various directions.
First, one could extend the mixture estimators to compute LOO-CV criteria for general scoring rules beyond the logarithmic one, see e.g.\ \citep{ber1979,vehtari2012survey}.
In such case one would be interested in LOO-CV estimators of quantities such as
$
\mathbb{E}_{y_{new}\sim p^*}[S(y_{new},p(\cdot|y))]
$
where $S$ is a scoring rule and $p(\cdot|y)$ is the predictive distribution of $y_{new}$ given the observed data $y$.
The main difference in terms of computational methodology that may arise is the need for another layer of integration if the scoring rule is not \emph{local} \citep{ber1979}, but instead defined itself as an integral.

Another important extension is to models with non conditionally independent observations, i.e.\ where the equality in  \eqref{eq:conditionally_ind} is not satisfied. 
There, the mixture distribution can be written as 
\begin{align*}
q_{mix}(\theta)=
Z^{-1}
\sum_{i=1}^n p(\theta)p(y_{-i}|\theta)  \propto p(\theta|y) \left(\sum_{i=1}^n p(y_{i}|\theta,y_{-i})^{-1}\right)\,, 
\end{align*}
but $p(y_{i}|\theta,y_{-i})\neq p(y_{i}|\theta)$ in general and thus the last equality in \eqref{eq:mixture} does not hold.
One should then replace $p(y_{i}|\theta)^{-1}$ with $p(y_{i}|\theta,y_{-i})^{-1}$ throughout for both the posterior and mixture estimators, e.g.\ in \eqref{eq:post_est}, \eqref{eq:mixture} and \eqref{eq:mix_est}.
In such contexts, the mixture estimators remain appealing \emph{provided} one can compute the $n$ predictive likelihood terms $\{p(y_{i}|\theta,y_{-i})\}_{i=1}^n$ for a given $\theta$ at $\Theta(n)$ total computational cost. This will be the case when, after the computation of $p(y|\theta)$, one can compute $p(y_{-i}|\theta)$ for a given $i$ at $\Theta(1)$ additional cost, e.g.\ using rank-one updates in regression-type models.
If instead computing each $p(y_{i}|\theta,y_{-i})$ term can only be done at $\Theta(n)$ cost separately for each $i$, then computing $\{p(y_{i}|\theta,y_{-i})\}_{i=1}^n$ has $\Theta(n^2)$ total cost and both the mixture and posterior estimators are likely to be impractical. 

Finally, another interesting direction to explore in future work is the extension of the proposed mixture estimator to leave-$p$-out contexts for $p>1$. A naive application of the mixture methodology, however, where the mixture is defined as
$q_{mix}(\theta)\propto \sum_{A} p(\theta)p(y_{-A}|\theta)$ where $A$ runs over subsets of $\{1,\dots,n\}$ of size $p$, would incur a $p$-choose-$n$ cost per iteration, thus being impractical. 
Nonetheless, we expect such cost to be avoidable using, for example, appropriate unbiased likelihood estimators in conjunction with pseudo-marginal MCMC algorithms. We leave such extensions to future work.

\subsection{Algorithmic variations}
As mentioned in Section \ref{sec:fin_var}, the mixture distribution $q_{mix}$ could be replaced by a more general, weighted version 
$q_{mix}^{(\balpha)}(\theta)=Z_{\balpha}^{-1}\sum_{i=1}^{n}\alpha_ip(y_{-i}|\theta)p(\theta)$
with $\boldsymbol{\alpha}=(\alpha_1,\dots,\alpha_n)\in(0,\infty)^n$ being arbitrary weights. In such case Theorem \ref{thm:mix_finite} would still hold, as shown in its proof.
In practice, such weighted version directly affects the value of the mixture weight components $(\pi_1,\dots,\pi_n)$ that in general satisfy $\pi_i\propto \alpha_i p(y_i|y_{-i})^{-1}$ for $i=\,\dots,n$, see also Remark \ref{rmk:mix}. 
Since larger values of $\pi_i$ tend to lead to estimators of $p(y_i|y_{-i})$ with smaller variance, it follows that increasing $\pi_i$ corresponds to putting more computational effort in estimating $p(y_i|y_{-i})$ relative to other $p(y_j|y_{-j})$ for $j\neq i$.
Thus, having direct control on $\pi_i$ might be useful to, e.g., design adaptive versions of the algorithm that adapt the weights $\boldsymbol{\alpha}$ on the fly to put more effort on more important or harder to estimate values of $p(y_i|y_{-i})$. 
In the default version, $\alpha_i=1$ and $\pi_i\propto p(y_i|y_{-i})^{-1}$ for $i=1,\dots,n$ . 
As discussed in Remark \ref{rmk:weights}, this is a reasonable default choice that gives more weight to data points $y_i$ with larger values of $|\log p(y_i|y_{-i})|$, which are typically more important (e.g.\ contribute more to LOO-CV) and harder to estimate. 
However, $\pi_i\propto p(y_i|y_{-i})^{-1}$ may not be the optimal choice in general, and thus weighted versions $q_{mix}^{(\balpha)}$ might be useful to increase robustness of the proposed estimating procedure to, e.g.\ overly large values of $\pi_i$.

As discussed in Remark \ref{rmk:mix}, the estimators  $\{\hat{\mu}_i^{(mix)}\}_{i=1}^n$
effectively estimate
the mixture weights $\{\pi_i\}_{i=1}^n$ and the normalizing constant $\tilde{Z}$ and then compute 
$p(y_i|y_{-i})=\tilde{Z}^{-1}\pi_i^{-1}$. 
One might consider more advanced methodologies, e.g.\ Bridge Sampling \citep{bennett1976efficient,meng1996simulating}, to estimate the normalizing constant $\tilde{Z}$ between $q_{mix}(\theta)$ and $p(\theta|y)$, but we expect this to lead to minimal improvements.
In fact, the largest relative variance in all our experiments was given by the estimators of $\pi_i$, i.e.\ the numerators in \eqref{eq:mix_est}, and thus we expect that employing a better estimator of $\tilde{Z}$ would only provide minimal improvements. 

\section{Discussion}\label{sec:disc}
We proposed a novel estimator for Bayesian LOO-CV estimator that retains appealing features of classical estimators, such as simplicity of implementation and $\Theta(Sn)$ total cost, while significantly improving robustness to high-dimensionality. We expect our proposed computational methodology to be most useful when the number of parameters is of comparable order, or even larger, than the number of data points or in the presence of highly influential data points.

Our work supports the idea that Bayesian LOO-CV computations can be efficiently accomplished with Monte Carlo methods, requiring a computational effort comparable to fitting the model once. 
This is a computational advantages compared to, e.g., marginal likelihood or Bayes Factors approximation, which are typically significantly harder tasks.

Directions for future research include characterizing how easy or hard it is to sample from $q_{mix}(\theta)$ compared to $p(\theta|y)$, which would provide a more complete theoretical picture on the comparison between the efficiency of classical and mixture estimators (see e.g.\ Remark \ref{rmk:multimod}); and extending the asymptotic analysis of Section \ref{sec:linear} to cases where both $n$ and $p$ diverge simultaneously.

%
%
%

\newpage
  \title{\bf Supplementary material for ``Robust leave-one-out cross-validation for high-dimensional Bayesian models''}
  \author{Luca Alessandro Silva\hspace{.2cm}\\
    Department of Decision Sciences, Bocconi University
\hspace{.2cm}\\    and Giacomo Zanella\hspace{.2cm}\\
    Department of Decision Sciences
    and BIDSA, Bocconi University
    }
\maketitle
\begin{abstract}
Section \ref{sec:impl} contains details on efficient and numerically stable implementations of sampling algorithms for the mixture distribution defined in equation \eqref{eq:mixture} of the main paper. Section \ref{sec:numerics_suppl} contains additional numerical experiments to integrate the ones in Section \ref{sec:sim} of the paper. Section \ref{sec:temp_est} contains theoretical and empirical results for the class of tempered estimators. Section \ref{proofs} contains mathematical proofs for the theoretical results stated in the paper.
\end{abstract}

\renewcommand\thesection{S.\arabic{section}}
\setcounter{section}{0}
\renewcommand{\theequation}{S.\arabic{equation}}
\setcounter{equation}{0}
\renewcommand{\thefigure}{S.\arabic{figure}}
\setcounter{figure}{0}
\section{Implementation details}\label{sec:impl}

\subsection{Sampling from the proposed mixture with MCMC}
\label{sec:sampling_mix_implem}
For models with conditionally independent data as in \eqref{eq:conditionally_ind} the log posterior is typically computed as the sum of log prior and log likelihood contributions as follows
\begin{align}\label{eq:log_post}
\log p(\theta|y)=
\log p(\theta) +\sum_{i=1}^n \log p(y_i|\theta)+const\,,
 \end{align}
where $const$ denotes terms that do not depend on $\theta$.
For $q_{mix}(\theta)$ defined in \eqref{eq:mixture} we have the same expression plus an additional term that can be written as follows to ensure numerical stability
\begin{align}\label{eq:log_mixture}
\log q_{mix}(\theta)=
\log p(\theta) +\sum_{i=1}^n \log p(y_i|\theta)
+
LSE(\{-\log p(y_i|\theta)\}_{i=1}^n)+
const\,,
\end{align}
where $LSE$ denotes the usual \emph{LogSumExp} function defined as $LSE(\textbf{x})=\log(\sum_{i=1}^n\exp(x_i))$ for $\textbf{x}=\{x_i\}_{i=1}^n\in\mathbb{R}^n$.
The expression in \eqref{eq:log_mixture} is trivial to compute whenever the log-prior and log-likelihoods are computable and requires $\Theta(n)$ operations per evaluation, exactly as $\log p(\theta|y)$. 
In other words, $q_{mix}(\theta)$  can be computed up to normalizing constant whenever the original posterior $p(\theta|y)$ can. 
We further note that, while computing $\log q_{mix}(\theta)$ in \eqref{eq:log_mixture} may appear to require roughly twice as many computations as $\log p(\theta|y)$ in \eqref{eq:log_post}, as one needs to compute both the sum and the $LSE$ quantities, for most models the cost of computing the $n$ likelihood terms $\{\log p(y_i|\theta)\}_{i=1}^n$ dominates the cost of computing their sum or the \emph{LSE} function, e.g.\ a $\Theta(np)$ cost for the former versus a $\Theta(n)$ cost for the latter for a regression model with $n$ data points and $p$ covariates. 
Thus in such cases computing $\log q_{mix}(\theta)$ and $\log p(\theta|y)$ have roughly the same cost.

The expression in \eqref{eq:log_mixture} is also trivial to differentiate, allowing to compute the gradient $\nabla \log q_{mix}(\theta)$, and is amenable to standard probabilistic programming software based automatic differentiation.
For example, for the logistic regression model used in Sections \ref{sec:leuk_stack} and \ref{sec:high_dim_bin} the \textsc{stan} code to define the posterior $ p(\theta|y)$ is given by
\begin{verbatim}
data {
    int <lower=0> n;
    int <lower=0> k;
    int <lower=0, upper=1> y[n];
    matrix [n,p] X;
    real <lower=0.0> prior_scale;
}
parameters {
    vector[p] beta;
}
model{
    vector[n] means=X*beta;
    target += double_exponential_lpdf(beta | 0, prior_scale);
    target += bernoulli_logit_lpmf(y | means);
}
\end{verbatim}
while to define the mixture distribution $q_{mix}(\theta)$ one should replace the \textrm{model} section with 
\begin{verbatim}
model{
    vector[n] means=X*beta;
    vector[n] log_lik;
    for (index in 1:n){
        log_lik[index]= bernoulli_logit_lpmf(y[index] | means[index]);
    }
    target += double_exponential_lpdf(beta | 0, prior_scale);
    target += sum(log_lik);
    target += log_sum_exp(-log_lik);
}
\end{verbatim}
See also 
\url{https://github.com/luchinoprince/Mixture_IS}
for more details and examples of software implementations.

\subsection{Efficient computation of mixture estimators}\label{sec:eff_comp_estim}
Given $S$ samples $\{\theta_s\}_{s=1}^S$ from $q_{mix}(\theta)$, the $n$ estimators $\{\hat{\mu}_i^{(mix)}\}_{i=1}^n$ defined in \eqref{eq:mix_est} can be computed at $\Theta(nS)$ total cost in a numerically stable way as follows:
\begin{itemize}
\item[(i)] compute the $n\times S$ matrix of log-likelihood terms $\{\ell_{is}\}_{i,s}$, where $\ell_{is}=\log p(y_i|\theta_s)$ for $i=1,\dots,n$ and $s=1,\dots,S$;
\item[(ii)] compute the $n\times S$ matrix of log-weights $\{\tilde{w}_{is}\}_{i,s}$ defined as $\tilde{w}_{is}=\log(w_i^{(mix)}(\theta_s))$, using the equality $\tilde{w}_{is}=-\ell_{is}-\tilde{z}_s$ for $i=1,\dots,n$ and $  s=1,\dots,S$, where $\tilde{z}_s=LSE(\{-\ell_{is}\}_{i=1}^n)$ for $s=1,\dots,S$;
\item[(iii)] compute the log-estimators exploiting the equality
$\log \hat{\mu}_i^{(mix)}=\tilde{z}-LSE(\{\tilde{w}_{is}\}_{s=1}^S)$ for $i=1,\dots,n$ where $\tilde{z}=LSE(\{-z_s\}_{s=1}^S)$.
\end{itemize}
The above operations (i)-(iii) require $\Theta(nS)$ computational cost. In terms of memory requirements, the simplest implementation of the above operations, which creates the $n\times S$ matrices $\{\ell_{is}\}_{i,s}$ and  $\{\tilde{w}_{is}\}_{i,s}$, require $\Theta(nS)$ storage, but this can be easily reduced to $\Theta(n)$ storage, if required, by storing only one column at a time.

\section{Additional Numerical experiments}\label{sec:numerics_suppl}
\subsection{Decomposition of MSE in Bias and Variance components}\label{sec:bias_var_dec}
Here we provide additional numerical illustrations on the behaviour of mean squared error (MSE) for the various estimators of $\log p(y_i|y_{-i})$ in the high-dimensional linear regression setting considered in Section \ref{sec:high_d_linregr_sims}. 
Relative to that section, we add to the comparison the bronze estimator and we decompose the MSE into bias squared and variance for all estimators. 
This allows to assess how each component contributes to the overall MSE in each case. Figure \ref{fig:bias:var} reports the results, considering an experimental setting completely analogous to the one with $\Sigma = 100/p \cdot\I_p$ in Section \ref{sec:high_d_linregr_sims} of the manuscript. 
\begin{figure}[t!]
\center
  \includegraphics[scale=0.35]{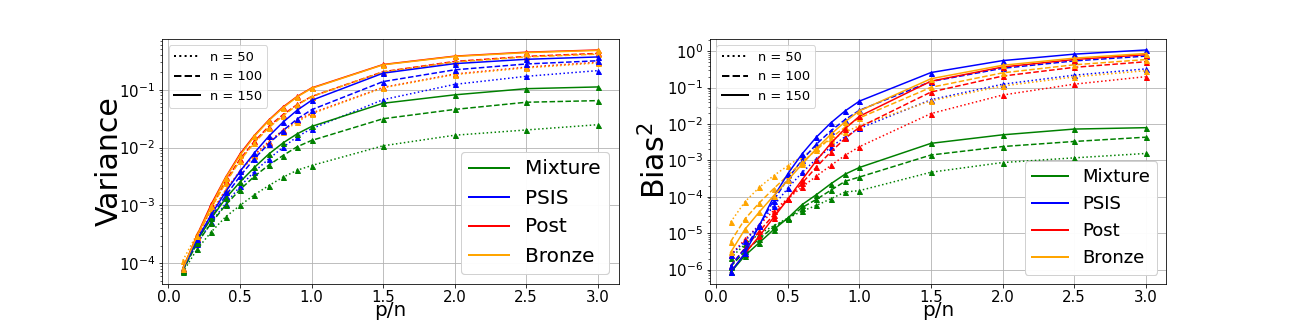}
   \caption{Variance (left) and Bias squared (right) for the posterior, PSIS, bronze and mixture estimators of $\{\log p(y_i|y_{-i})\}_{i=1}^n$ for the high-dimensional linear regression setting of Section \ref{sec:high_d_linregr_sims} with $\Sigma = 100/p\cdot\I_p$. 
   See Section \ref{sec:bias_var_dec} for more details.
}
   \label{fig:bias:var}
 \end{figure}
In particular we have $\sigma^2=1$, each estimator is obtained through $2\times 10^3$ i.i.d.\ samples and values are averaged over $10^4$ replicates (i.e.\ $10^4$ random dataset for each $n$ and $p$ combination). 
We can see from Figure \ref{fig:bias:var} that mixture estimators have small bias squared, and their MSE is  indeed dominated by their variance. 
Both the bias and variance of other estimators under consideration are significantly larger, with their bias squared being particularly large in high-dimensional settings.

\subsection{MSE figure for the Bladder dataset}
Figure \ref{fig:bladder} illustrate the MSE, averaging over $i=1,\dots,n$, for the experiments reported in Table \ref{table:bladder} of the paper.
\begin{figure}[h!]
\centering
\includegraphics[width=0.65\linewidth]{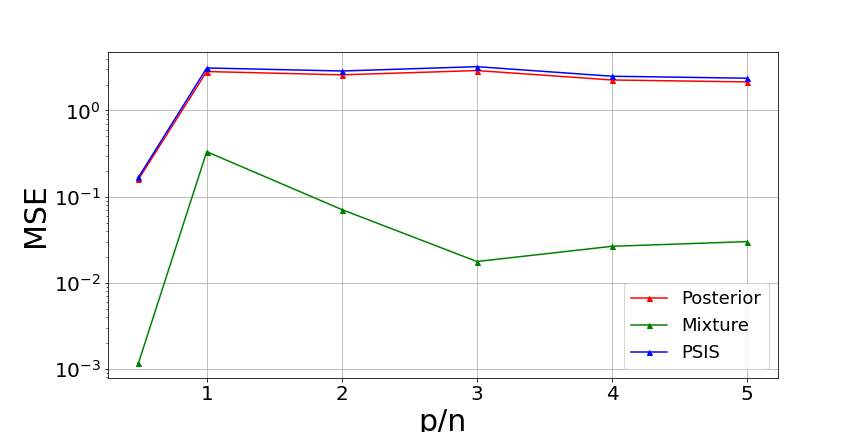}
\caption{Average MSE on different sub-datasets of the Bladder Cancer data. See Table \ref{table:bladder} of the paper for more details.}\label{fig:bladder}
\end{figure}

\subsection{Traceplots for high-dimensional binary regression}
Figure \ref{fig:traces} displays the evolution of the classical and mixture estimators for the 20 data points with largest absolute value of $\log (p(y_i|y_{-i}))$, for the examples reported in Table \ref{table:high_d} of the paper.
Some classical estimators exhibit very large jumps even at high number of iterations, which is a typical pathological behaviour of estimators with infinite or excessively large variance.
The mixture estimators, despite having some jumps in a few cases, display a much more stable evolution and convergence.
\begin{figure}[h!]
\begin{subfigure}{.32\textwidth}
\centering
\includegraphics[scale=0.18]{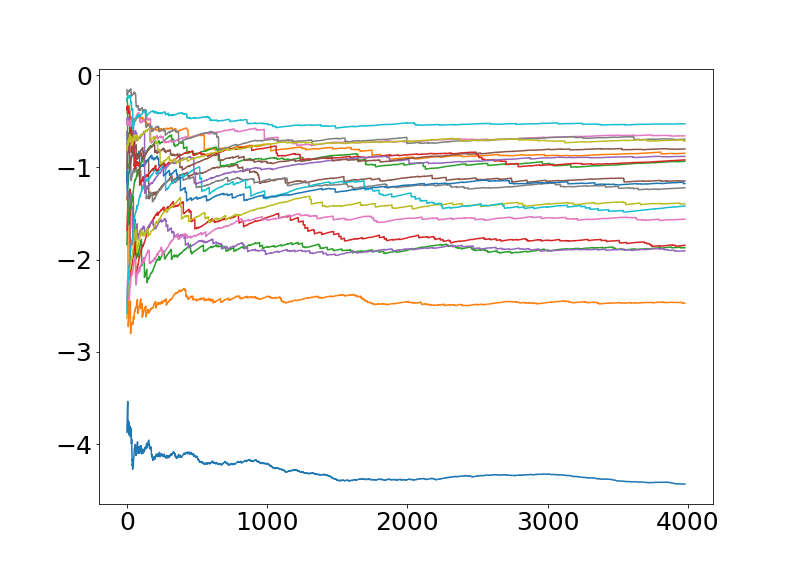}
\end{subfigure}
\begin{subfigure}{.32\textwidth}
\includegraphics[scale=0.18]{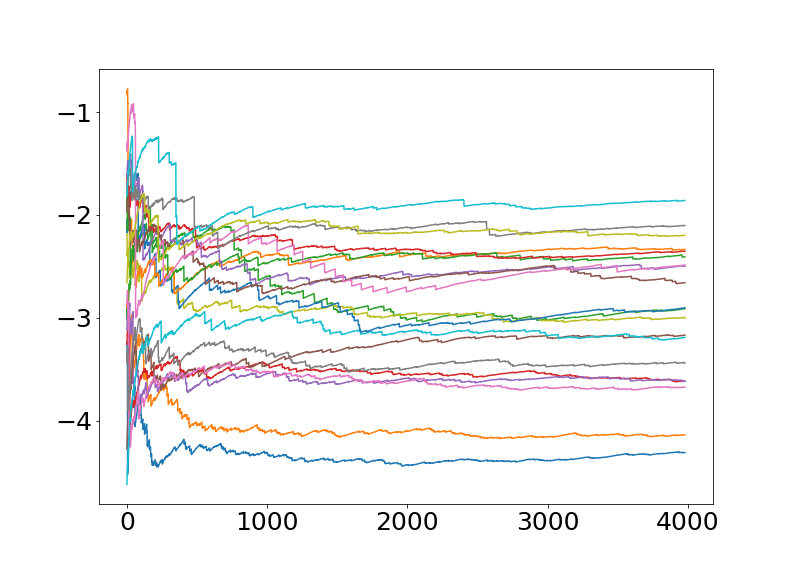}
\end{subfigure}
\begin{subfigure}{.32\textwidth}
\includegraphics[scale=0.18]{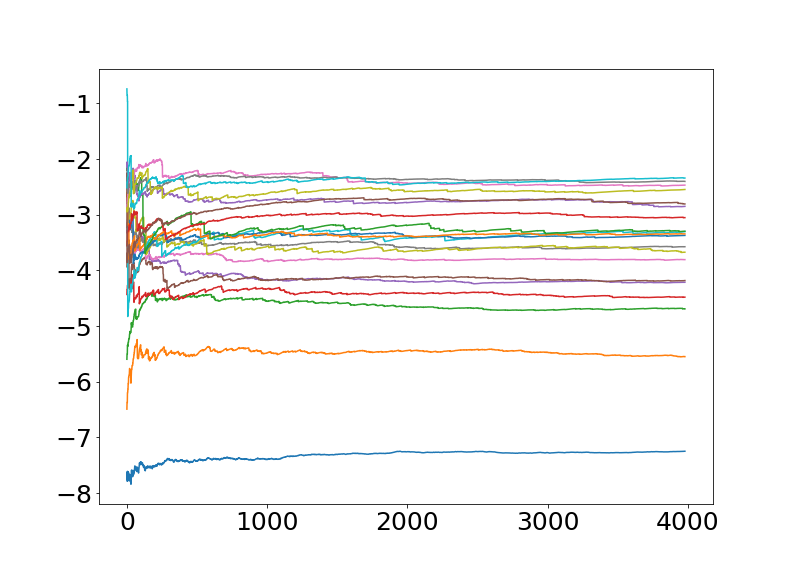}
\end{subfigure}
\newline
\begin{subfigure}{.32\textwidth}
\centering
\includegraphics[scale=0.18]{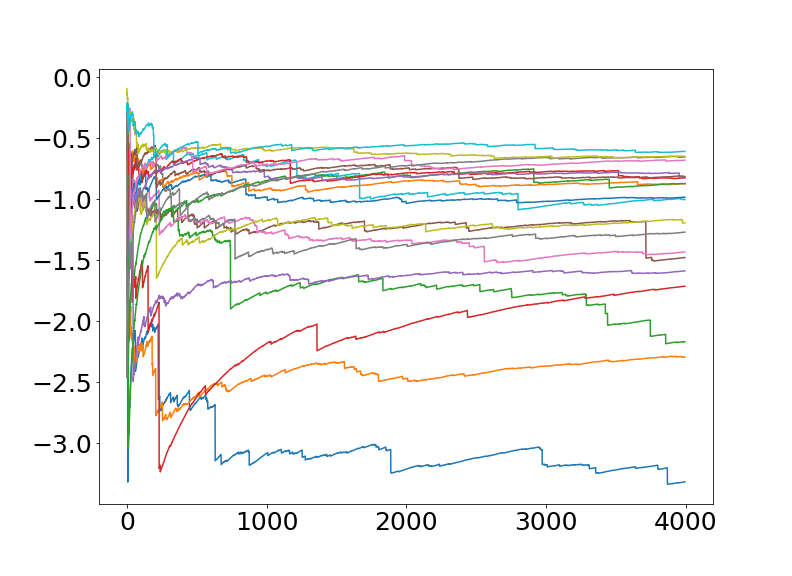}
\caption{Voice}
\end{subfigure}
\begin{subfigure}{.32\textwidth}
\includegraphics[scale=0.18]{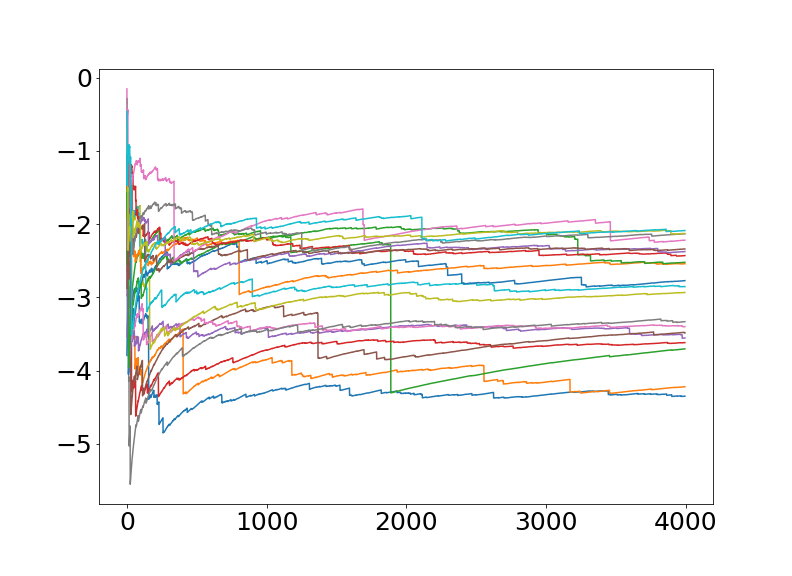}
\caption{Parkinson}
\end{subfigure}
\begin{subfigure}{.32\textwidth}
\includegraphics[scale=0.18]{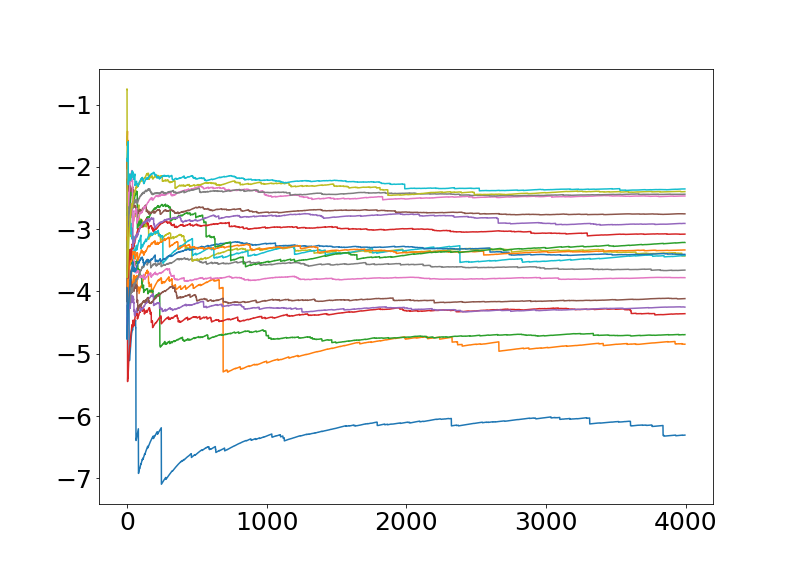}
\caption{Arrhythmia}
\end{subfigure}
\caption{Evolution of the mixture (first row) and classical (second row) estimators, with number of samples on the $x$-axis, for three datasets (one per column). The traceplots of the estimators corresponding to the 20 data points with largest absolute value of $\log (p(y_i|y_{-i}))$ are displayed.}
\label{fig:traces}
\end{figure}

\subsection{Examples from the Bayesian LOO-CV literature}\label{sec:leuk_stack}
In this section we consider the \emph{Leukaemia} and \emph{Stack Loss} datasets, which are standard example in the literature on Bayesian LOO-CV computation 
\citep{per1997,epi2008,vehtari2016,rischard2018unbiased}.
The first dataset is used to estimate the survival distribution for leukaemia patients. 
The response variable is survival time (from diagnosis), and the two explanatory variables are white blood cell count at diagnosis (WBC) and the outcome of a test related to white blood cell characteristics \cite{cook1982}.
Following previous analysis in the literature, we dichotomize survival times 
to indicate survival past 50 weeks, and we discard three repeated observation. The resulting dataset has $n=30$ binary responses, $p=3$ regressors including the intercept and is available at \url{https://github.com/luchinoprince/Mixture_IS}. 
We fit a Bayesian logistic regression model, meaning that each response $y_i\in\{0,1\}$ 
is modelled as a Bernoulli random variables taking value $1$ with success probability $(1 + \exp(x_i^T\theta))^{-1}\exp(x_i^T\theta)$, where $x_i$ is a vector of covariates.
We assume independent Laplace, or double-Exponential, priors for $\theta_1,\dots,\theta_p$ with mean parameter equal to $0 $ and scale parameter equal to $b=\sqrt{50/p}$, so to have prior variance for each coefficient equal to $100/p$.

This dataset is challenging for LOO-CV calculations due to the presence of a highly-influential observation, a patient with a high WBC and a survival time of more than 50 weeks, here corresponding to $i=15$.
In particular, \citep{epi2008} show that for this dataset $AV_{15}^{(post)}=\infty$, while we know by Theorem \ref{thm:mix_finite} that $AV_{15}^{(mix)}<\infty$. 

The values of $\{ \mu_i\}_{i=1}^n$, where $\mu_i=p(y_i|y_{-i})$, are not available analytically, and we compute accurate approximations of them running a separate long MCMC chain to sample from $p(\theta|y_{-i})$, for each $i=1,\dots,n$, with $10^6$ iterations and first half discarded as burn in.
We treat such estimates as ground truth values, since their Monte Carlo error is negligible compared to the ones of the other estimators involved in this analysis.
We then run $100$ independent MCMC chains sampling from $p(\theta|y)$ and from $q_{mix}(\theta)$, of length  $2\times 10^4$ iterations each with the first half discarded as burn-in, and use the resulting samples to compute 100 i.i.d.\ replicates of the classical, mixture and PSIS estimators. 
All MCMC runs were obtain with the \textsc{stan} interface in \textsc{python}, see e.g.\ \url{https://pystan.readthedocs.io/}, using default settings, see Section \ref{sec:impl} for detail on how to sample from $q_{mix}$ with \textsc{stan}. No convergence or mixing issues were found using standard diagnostics.

Figure \ref{fig:leuk} reports the results displaying, for each $i=1,\dots,n$, a box-plot of the differences between the log probability $\log \mu_i$ and its 
100 estimates.
As we can see, the classical and PSIS estimators struggle to recover the true value of $\log \mu_i$ for $i=15$ providing highly biased estimates, which is in line with the results of \cite{epi2008,vehtari2016}.
On the contrary, mixture estimators have drastically smaller errors and are centred around the correct values. 
All methods are able to accurately recover the ground truth values for the other values of $i$.
\begin{figure}[t!] 
\includegraphics[scale=0.35]{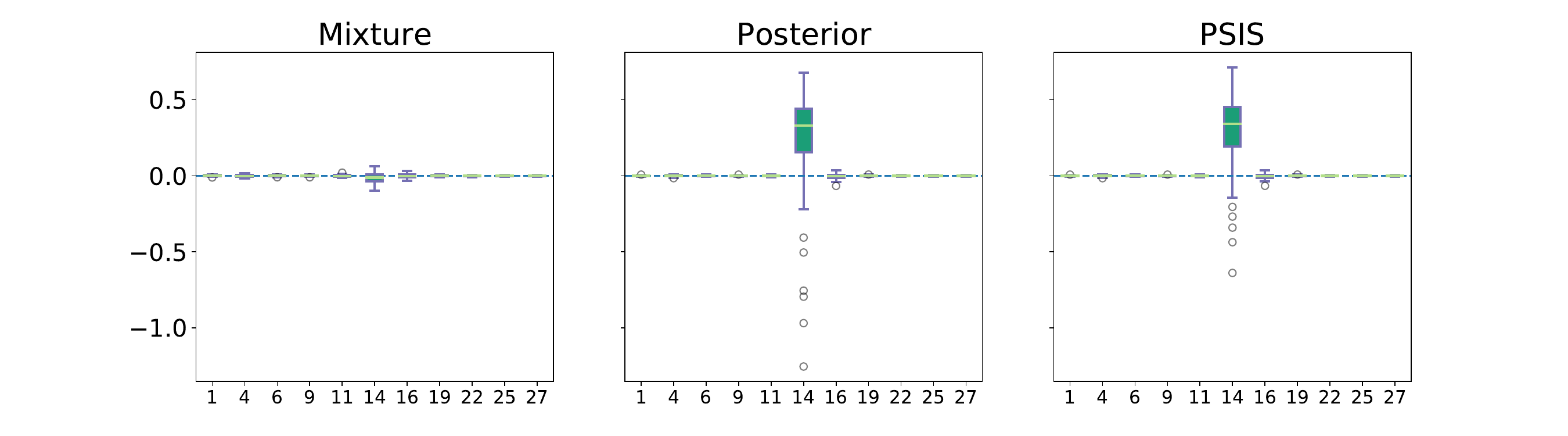}
 \caption{Values of $\log \hat\mu_i-\log \mu_i$ ($y$-axis) across $i\in\{1,\dots,n\}$ ($x$-axis) over $100$ repetitions for the Leukaemia dataset $(n=30)$, 
where $\hat\mu_i$ is either $\hat\mu^{(mix)}_i$ (left),  $\hat\mu^{(post)}_i$ (center) or $\hat\mu^{(psis)}_i$  (right). 
See Section \ref{sec:leuk_stack} for more details.
}
 \label{fig:leuk}
 \end{figure}


We now consider a second dataset previously analysed in the Bayesian LOO-CV literature, namely 
the \emph{Stack Loss} dataset as in \citet{per1997} and \citet[Section 4.3]{vehtari2016}, obtaining a linear regression model with $n=21$ observations and $p=3$ regressors. 
For this example \citet{per1997} shows $AV_i^{(post)}=\infty$ for $i=21$.
Figure \ref{fig:stack} displays the root mean squared error (RMSE) in estimating $\log(p(y_i|y_{-i}))$ for the problematic observation, $i=21$, as well as a more ordinary observation, $i=1$.
We fit the model with different values of $\sigma^2$, varying them over a grid centred on the maximum marginal likelihood estimator, in order to explore sensitivity to the likelihood strength.
In this example PSIS improves over the posterior estimator for both $i=1$ and $i=21$.
For both posterior and PSIS, the RMSE for $i=21$ is an order of magnitude larger than the one for $i=1$, while for the mixture they are of comparable order. 
As a result, mixture estimators provides a major improvement for $i=21$, while it performs comparably for $i=1$.
This relates to the fact that mixture estimators implicitly focus more computational effort on smaller and harder to estimate values of $p(y_i|y_{-i})$ (see e.g.\ Remark \ref{rmk:weights} in the paper), and thus they are particularly useful for those.
\begin{figure}[t!]
\centering
\includegraphics[width=0.95\linewidth]{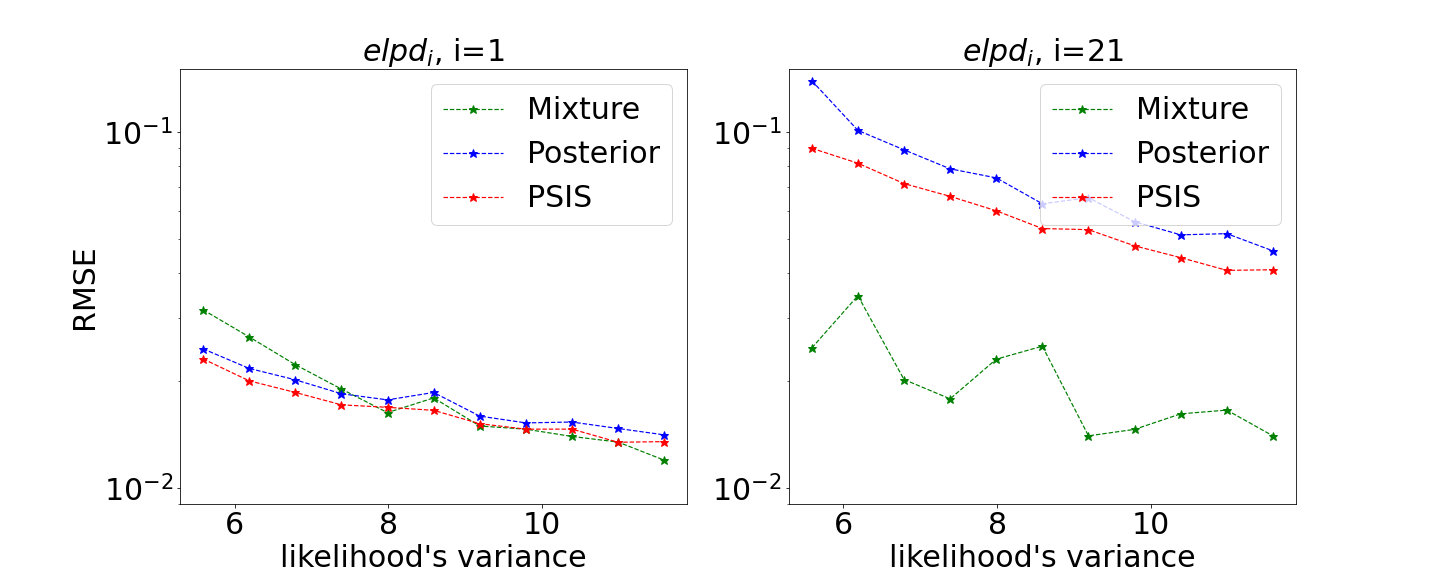}
\caption{Root mean squared error (RMSE) in estimating $\log(p(y_i|y_{-i}))$ for the \emph{Stack Loss} data for $i=1$ (left) and $i=21$ (right).
The x-axis reports the value of $\sigma^2$. See Section \ref{sec:leuk_stack} for details.}\label{fig:stack}
\end{figure}

\color{black}

\section{Comparison with geometric tempering estimators}
\label{sec:temp_est}
In this section we analyse a class of estimators based on geometric tempering, of which $\hat\mu_i^{(post)}$ and $\hat\mu_i^{(br)}$ are particular instances. First we define the class of estimators.
\begin{definition}
We define the $\alpha$-tempered posterior as 
\begin{align}
\label{temp:post}
q_{(tmp,\alpha)}(\theta)&\propto p(y|\theta)^\alpha p(\theta)&\alpha\in[0,1]
\end{align}
and the $\alpha$-tempered estimators $\{\hat\mu_i^{(tmp,\alpha)}\}_{i=1,\dots,n}$ as the self-normalised importance sampling estimators obtained using \eqref{temp:post} as importance distribution and $p(\theta|y_{-i})$ as target.
We denote the asymptotic variance of such estimators as 
\begin{align}
 AV^{(tmp,\alpha)}_i&=\lim_{S\to\infty}S\,\var (\hat{\mu}_i^{(tmp,\alpha)}/\mu_i)&i=1,\dots,n\,.
\end{align}
\end{definition}
Note that $q_{(1,tmp)}(\theta)=p(\theta|y)$ and $q_{(1-1/n,tmp)}(\theta)=q_{br}(\theta)$ with $q_{br}$ as in \eqref{bronze} of the main paper, which implies that $\hat\mu_i^{(1,tmp)}\equiv \hat\mu_i^{(post)}$ and $\hat\mu_i^{(1-1/n,tmp)}\equiv \hat\mu_i^{(br)}$, where the latter denotes bronze estimators of \citep{alqallaf2001cross} discussed in Section \ref{sec:alternatives} of the paper. 

We consider the Gaussian regression model 
\begin{equation}
\label{modelcov}
\begin{gathered}
y_i|\theta \sim N(x_i^T\theta,\sigma^2)\qquad i=1,\dots,n\\
\theta\sim N(\theta_0,\Sigma)\,,\qquad\qquad\qquad
\end{gathered}
\end{equation}
although we expect the qualitative behaviour discussed here to hold similarly also in more general settings. 
We have the following theorem, whose proof can be found below in Section \ref{proofs}.
\begin{theorem}
\label{tmp:lev}
Under \eqref{modelcov}, for each $i\in\{1,\dots,n\}$, we have $AV_{i}^{(tmp, \alpha)}<\infty$ if and only if $H^{(\alpha)}_{ii}<\frac{2}{2-\alpha}$, where 
\begin{equation}
H^{(\alpha)} = X\left(X^TX+\frac{\sigma^2}{2-\alpha}\Sigma^{-1}\right)^{-1}X^T.  
\end{equation}
\end{theorem}
Theorem \ref{tmp:lev} suggests that lowering the tempering parameter helps in increasing the leverage value at which the estimators have infinite asymptotic variance (see e.g.\ \citep{walker1988influence} and related discussion after Theorem \ref{thm:leverage} in the paper). 
 In particular, setting $\alpha=0$ would produce an estimator that is guaranteed to have finite asymptotic variance since leverages are always bounded by $1$. However, this would reduce to using the prior itself as importance distribution, which is well-known to be a poor choice in most commonly encountered settings (i.e.\ the induced asymptotic variance will be finite but very large). We illustrate this numerically in Figure \ref{temp:mse}. We consider the model in \eqref{modelcov} with $(n,p)=(100,100)$, $\sigma^2=1$ and $\Sigma=100/p\cdot \I_p$, and generate $10^3$ synthetic datasets as in Section \ref{sec:high_d_linregr_sims}.  We test the performance of different tempering estimators, each generated from $2\times 10^3$ i.i.d.\ samples from $q_{(tmp,\alpha)}$.
\begin{figure}
\center
\includegraphics[scale=0.4]{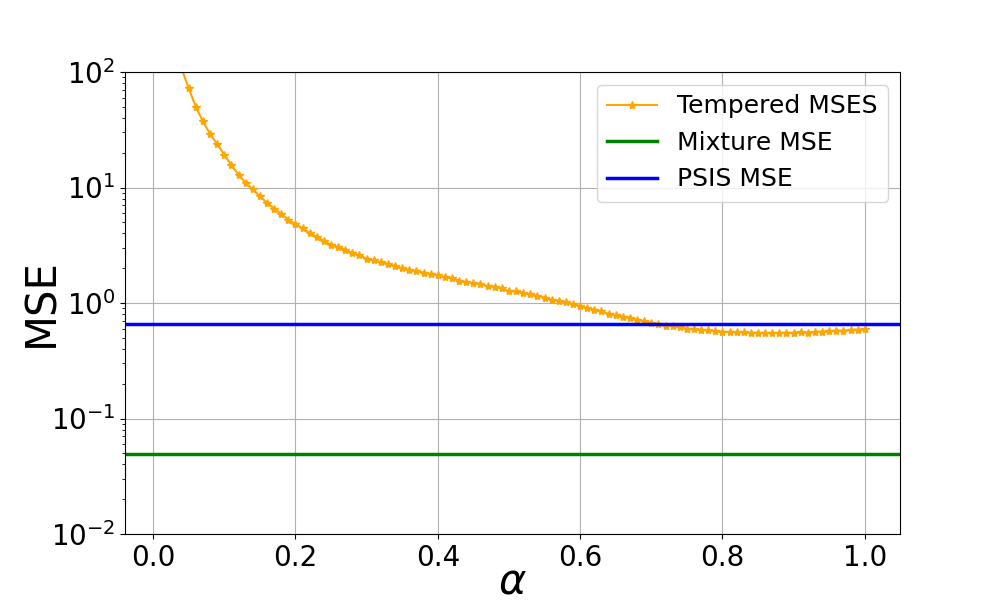}
\caption{Mean squared error (MSE) in estimating $\log(p(y_i|y_{-i}))$ for different tempered estimators as a function of the temperature parameter $\alpha$, for a high-dimensional linear regression model with
$n=p=100$. Each estimator is computed using $2\times 10^3$ i.i.d.\ samples from either $q_{(tmp,\alpha)}$ or $q_{mix}$. See Section \ref{sec:temp_est} for details.}
\label{temp:mse}
\end{figure}
From Figure \ref{temp:mse} we can see that tempering only provides mild improvements over classical posterior sampling (which is the special case $\alpha=1)$ and PSIS, and starts performing very poorly when $\alpha$ gets small. 
 See also Figure \ref{fig:bias:var} for more extensive results on the bias and variance of bronze estimators.

\section{Proofs}\label{proofs}
\begin{proof}[Proof of Theorem \ref{thm:mix_finite}]
A standard application of the delta method for the derivation of the relative asymptotic variance of self-normalized importance sampling estimators, see e.g.\ 
\citet[eq.(2.7)]{liu2001monte} or \citet[eq.(9.8)]{owen}, applied to $\hat{\mu}_i^{(mix)}$ leads to
\begin{align}
AV_i^{(mix)}=&
\lim_{S\to\infty}S\,\var \left(\frac{\hat{\mu}_i^{(mix)}}{\mu_i}\right)
=\int
\left(\frac{p(\theta|y_{-i})}{q_{mix}^{(\balpha)}(\theta)}\right)^2
\left(\frac{p(y_i|\theta)}{\mu_i}-1\right)^2
q_{mix}^{(\balpha)}(\theta)d\theta
\nonumber\\
=&\int\frac{p(\theta|y)^2}{q_{mix}^{(\balpha)}(\theta)}d\theta
-2\int\frac{p(\theta|y)p(\theta|y_{-i})}{q_{mix}^{(\balpha)}(\theta)}d\theta
+\int\frac{p(\theta|y_{-i})^2}{q_{mix}^{(\balpha)}(\theta)}d\theta\,,\label{avmix}
\end{align}
where in the last equality we re-arranged terms and used $\mu_i^{-1}p(\theta|y_{-i})p(y_i|\theta)=p(\theta|y)$.
Writing $q_{mix}^{(\balpha)}(\theta)=\sum_{j=1}^{n}\pi_jp(\theta|y_{-j})$ with $\pi_j=Z_{\balpha}^{-1}\alpha_jp(y_{-j})$ and upper bounding the negative terms in \eqref{avmix} by $0$, we have
\begin{equation}\label{eq:av_mix_bound}
AV_i^{(mix)}\leq
\int\frac{p(\theta|y)^2}{\sum_{j=1}^{n}\pi_jp(\theta|y_{-j})}d\theta
+\int\frac{p(\theta|y_{-i})^2}{\sum_{j=1}^{n}\pi_jp(\theta|y_{-j})}d\theta.
\end{equation}
From $\sum_{j=1}^{n}\pi_jp(\theta|y_{-j})\geq \pi_ip(\theta|y_{-i})$  it follows
\[
\int\frac{p(\theta|y_{-i})^2}{\sum_{j=1}^{n}\pi_j p(\theta|y_{-j})}d\theta
\leq
\int\frac{p(\theta|y_{-i})^2}{\pi_ip(\theta|y_{-i})}d\theta
=\pi_i^{-1}
\]
and 
\[
\int\frac{p(\theta|y)^2}{\sum_{j=1}^{n}\pi_jp(\theta|y_{-j})}d\theta
\leq
\pi_i^{-1}
\int\frac{p(\theta|y)^2}{p(\theta|y_{-i})}d\theta
=\pi_i^{-1}p(y_i|y_{-i})^{-1}
\int p(y_i|\theta)p(\theta|y)d\theta\,,
\]
where in the last equality we also used 
$
p(\theta|y_{-i})^{-1}
p(\theta|y)
=p(y_i|y_{-i})^{-1}p(y_i|\theta)
$.
Combining the above with \eqref{eq:av_mix_bound} we obtain
\begin{equation}\label{eq:av_mix_bound_2}
AV_i^{(mix)}\leq
\pi_i^{-1}\left(
1+p(y_i|y_{-i})^{-1}
\int p(y_i|\theta)p(\theta|y)d\theta
\right)\,.
\end{equation}
The latter upper bound is finite by \eqref{eq:reg_assumptions} and the fact that $\alpha_i>0$ implies $\pi_i>0$.

Finally, \eqref{eq:MSE_log_scale} follows from the usual bias-variance decomposition combined with $AV_i^{(mix)}<\infty$ and the fact that $\hat{\mu}_i^{(mix)}$, being a self-normalized importance sampling estimator finite asymptotic variance, has $\mathcal{O}(S^{-1})$ bias as $S\to\infty$. The latter is a well-known fact of which we provide a proof for completeness.  
Recall that 
\begin{align*}
\hat{\mu}_i^{(mix)}=\frac{A_S}{B_S},\quad A_S=S^{-1}\sum_{s=1}^{S}p(y_i|\theta_s)\frac{p(\theta_s|y_{-i})}{q_{mix}^{(\balpha)}(\theta_s)},\quad
B_S=S^{-1}\sum_{s=1}^{S}\frac{p(\theta_s|y_{-i})}{q_{mix}^{(\balpha)}(\theta_s)}
\end{align*}
where $\theta_1,\dots,\theta_S$ are i.i.d.\ samples from $q_{mix}^{(\balpha)}$.
It is easy to see that $\E[A_S]=\mu_i$ and $\E[B_S]=1$. Also, $\var(A_S/\mu_i)=c_1/S$ and $\var(B_S)=c_2/S$, with $c_1=\int\frac{p(\theta|y)^2}{q_{mix}^{(\balpha)}(\theta)}d\theta<\infty$ and $c_2=\int\frac{p(\theta|y_{-i})^2}{q_{mix}^{(\balpha)}(\theta)}d\theta<\infty$ by $AV_i^{(mix)}<\infty$ and \eqref{avmix}.
Thus, by the delta method, we have
\begin{align*}
\lim_{S\to\infty}S|\E[\log(\hat{\mu}_i^{(mix)})-\log(\mu_i)]|
&\leq
\lim_{S\to\infty}S|\E[\log(A_S/\mu_i)
]|
+
\lim_{S\to\infty}S|\E[\log(B_S)]|\\
&=(c_1+c_2)/2<\infty\,.
\end{align*}
The latter implies an $\mathcal{O}(S^{-1})$ bias and thus the statement in \eqref{eq:MSE_log_scale}.
\end{proof}

\subsection{Proof of Theorem \ref{thm:leverage}}
\begin{lemma}
\label{lemma:a1}
For any $h\in(0,1)$, the matrix $[X^TX+\sigma^2\cdot\Sigma^{-1}-h^{-1}x_ix_i^T]$ is singular if and only if $H_{ii}=h$, with $H$ as in \eqref{eq:hat_mat}. If $H_{ii}\neq h$ then
\begin{multline}
\label{invh}
 [X^TX+\sigma^2\cdot\Sigma^{-1}-h^{-1}x_ix_i^T]^{-1}= \\ (X^TX)^{-1}+\frac{h^{-1}}{1-h^{-1}H_{ii}}\cdot (X^TX+\sigma^2\Sigma^{-1})^{-1}x_ix_i^T(X^TX+\sigma^2\Sigma^{-1})^{-1}\,.
\end{multline}
\end{lemma}
\begin{proof}
Assume first that $H_{ii}=h$, then $x_i\neq 0$ (the zero vector has leverage zero). Multiplying $[X^TX+\sigma^2\cdot\Sigma^{-1}-h^{-1}x_ix_i^T]$ by the non-zero vector $(X^TX+\sigma^2\Sigma^{-1})^{-1}x_i$ yields $[X^TX+\sigma^2\cdot\Sigma^{-1}-h^{-1}x_ix_i^T](X^TX+\sigma^2\Sigma^{-1})^{-1}x_i= x_i+h^{-1}H_{ii}x_i=0$. Hence we have proved that in this case $[X^TX+\sigma^2\cdot\Sigma^{-1}-h^{-1}x_ix_i^T]$ is singular. We now verify that $\eqref{invh}$ is the inverse of $[X^TX+\sigma^2\cdot\Sigma^{-1}-h^{-1}x_ix_i^T]$ when $H_{ii}\neq h$. Multiplying the two matrices we get:
\begin{multline*}
\I+\frac{h^{-1}}{1-h^{-1}H_{ii}}x_ix_i^T(X^TX+\sigma^2\cdot\Sigma^{-1})^{-1}-h^{-1}x_ix_i^T(X^TX+ \sigma^2\cdot\Sigma^{-1})^{-1}\\ 
+\frac{h^{-2}H_{ii}}{1-h^{-1}H_{ii}}x_ix_i^T(X^TX+\sigma^2\cdot\Sigma^{-1})^{-1}
\end{multline*}
\[= \I+x_ix_i^T(X^TX+\sigma^2\cdot\Sigma^{-1})^{-1}\Big(\frac{h^{-1}}{1-h^{-1}H_{ii}}+h^{-1}-\frac{h^{-2}H_{ii}}{1-h^{-1}H_{ii}}\Big) =\I\,.\]
\end{proof}

\begin{lemma}
\label{hii_bayes}
Let $\Sigma$ be a positive definite $p\times p$ matrix, $X$ a $n\times p$ matrix, $\sigma>0$ and $M=X^TX-h^{-1}x_ix_i^{T}+\sigma^2\Sigma^{-1}$ with $i\in\{1,\dots,n\}$ and $h\in(0,1)$.
Then $M$ is positive definite if and only if $H_{ii}<h$, with $H$ as in \eqref{eq:hat_mat}.
\end{lemma}
\begin{proof}
Assume that $[X^TX+\sigma^2\cdot\Sigma^{-1}-h^{-1}x_ix_i^T]$ is positive definite. If $x_i=0$, then $H_{ii}=0<h$. If $x_i\neq 0$, then $(X^TX+\sigma^2\Sigma^{-1})^{-1}x_i$ is a non-zero vector and we must have, by positive definiteness,  $0<x_i^T(X^TX+\sigma^2\Sigma^{-1})^{-1}[X^TX+\sigma^2\cdot\Sigma^{-1}-h^{-1}x_ix_i^T](X^TX+\sigma^2\Sigma^{-1})^{-1}x_i=H_{ii}(1-h^{-1}H_{ii})$. This implies that $H_{ii}<h$.

Conversely, suppose that $H_{ii}<h$. Then $h^{-1}/(1-h^{-1}H_{ii})>0$, and $\eqref{invh}$ shows that $[X^TX+\sigma^2\cdot\Sigma^{-1}-h^{-1}x_ix_i^T]^{-1}$ can be written as the sum of a positive definite matrix and a positive semi-definite one. As such, it is positive definite, and $[X^TX+\sigma^2\cdot\Sigma^{-1}-h^{-1}x_ix_i^T]$ must be positive definite as well.
\end{proof}

\begin{proof}[Proof of Theorem \ref{thm:leverage}]
The relative asymptotic variance of the classical estimator $\hat{\mu}^{(post)}_i$ can be derived in analogous way to the derivation in \eqref{avmix}, with the importance distribution $q_{mix}^{(\balpha)}(\theta)$ replaced by the posterior $p(\theta|y)$. After simplifications, this leads to 
\begin{equation}
\label{varpost}
 AV^{(post)}_i=\int\left(\frac{p(\theta|y_{-i})}{p(\theta|y)}\right)^2p(\theta|y)d\theta-1.
\end{equation}
By \eqref{varpost} and \eqref{modelcov} we have
\begin{align}
AV^{(post)}_i+1
&=c_1
\int exp\Big\{ \frac{2(y_i-x_i^T\theta)^2}{2\sigma^2}-\frac{(y-X\theta)^T(y-X\theta)}{2\sigma^2}-\frac{(\theta-\theta_0)^T\Sigma^{-1}(\theta-\theta_0)}{2}\Big\} d\theta,
\end{align}
where $c_1$ is a constant independent of $\theta$.
Grouping together quadratic and linear terms in $\theta$ we obtain 
\begin{align}
AV^{(post)}_i
&=c_2\int exp\Big\{-\theta^TM\theta
+\theta^Tv
\Big\}d\theta\,,\quad
M=(2\sigma^2)^{-1}[X^TX-2x_ix_i^T+\sigma^2\Sigma^{-1}]
\end{align}
where $v$ is a $p$-dimensional vector
and $c_2$ is a non-zero scalar, and both are independent of $\theta$.
It follows that $AV^{(post)}_i$ is finite if and only if $M$ is positive definite.
The thesis follows by Lemma \ref{hii_bayes} with $h=0.5$.
\end{proof}

\subsection{Proof of Theorem \ref{tmp:lev}}
\begin{proof}[Proof of Theorem \ref{tmp:lev}]
The proof is analogous to that of Theorem \ref{thm:leverage}. First, the same arguments as in \eqref{avmix} give 
\begin{equation}
\label{av:temp}
AV_i^{(tmp,\alpha)}=\int\frac{p(\theta|y)^2}{q_{(tmp,\alpha)}(\theta)}d\theta
-2\int\frac{p(\theta|y)p(\theta|y_{-i})}{q_{(tmp,\alpha)}(\theta)}d\theta
+\int\frac{p(\theta|y_{-i})^2}{q_{(tmp,\alpha)}(\theta)}d\theta.
\end{equation}
The first two integrals in \eqref{av:temp} can be shown to be finite combining basic manipulations with the fact that $\int p(y|\theta)^\alpha p(\theta)d\theta\in(0,\infty)$ for every $\alpha\in(0,\infty)$ and that $p(y_{-i})>0$.
Hence the finiteness of \eqref{av:temp} depends on the behaviour of the last integral and 
 we have
\begin{equation}
\label{av:br}
AV^{(tmp, \alpha)}_i = c_2 +
\int\left(\frac{p(\theta|y_{-i})}{q_{(tmp,\alpha)}(\theta)}\right)^2q_{(tmp, \alpha)}(\theta)d\theta = c_2+c_1\int \frac{p(y_{-i}|\theta)^{2-\alpha}}{p(y_i|\theta)^{\alpha}}d\theta,
\end{equation}
where $c_1,c_2$ are positive and finite constants that do not depend on $\theta$. 
Grouping together quadratic and linear terms in $\theta$ we obtain
\begin{equation}
AV^{(tmp,\alpha)}_i =c_2 + \tilde c_1 \int \hbox{exp}\Big\{-\theta^TM\theta+\theta^Tv\Big\}d\theta,
\end{equation}
where
\begin{equation}
M=\frac{2-\alpha}{2\sigma^2} \Big[X^TX-\frac{2}{2-\alpha} x_i^Tx_i+\frac{\sigma^2}{2-\alpha}\Sigma^{-1}\Big],
\end{equation}
$v$ is a $p$-dimensional vector
and $\tilde c_1$ and $c_2$ are non-zero scalars (both independent of $\theta$).
It follows that $AV^{(tmp, \alpha)}_i$ is finite if and only if $M$ is positive definite.
The thesis follows by Lemma \ref{hii_bayes} with $h=1-\alpha/2$.
\end{proof}

\subsection{Proof of Proposition \ref{prop:high_p} and Theorems \ref{thm:limsup_fin}, \ref{thm:limsup_fin_general} and \ref{thm:av_loo_general}}
\begin{proof}[Proof of Proposition \ref{prop:high_p}]
Denoting $\lambda_p=\nu_p^{-2}\sigma^2$ and applying Woodbury matrix identity we have
\begin{align*}
H
&=X(X^TX+\lambda_p\I_p)^{-1}X^T\,,\\
&=X(\lambda_p^{-1}\I_p-\lambda_p^{-1}X^T(\I_n+\lambda_p^{-1} XX^T)^{-1}\lambda_p^{-1}X )X^T\\
&=\lambda_p^{-1}XX^T
-\lambda_p^{-1}XX^T(\I_n+\lambda_p^{-1} XX^T)^{-1}\lambda_p^{-1}XX^T\,.
\end{align*}
Since $\lim_{p\to\infty}p\lambda_p^{-1}=\sigma^{-2}c$, by Kolmogorov’s criterion of SLLN and the random design assumption \eqref{eq:rand_des} on $X$, we have $p^{-1}XX^T\to \tau^2\I_n$ almost surely element-wise as $p\to\infty$, see e.g.\ \citet[Lemma 1]{fasano2019scalable} for a more detailed proof of the latter statement. 
It follows that $\lambda_p^{-1}XX^T\to \frac{c\tau^2}{\sigma^2}\I_n$ almost surely element-wise as $p\to\infty$ and $H$ converges in the same way to
\[ \frac{c\tau^2}{\sigma^2}\I_n-\frac{c\tau^2}{\sigma^2}\I_n(\I_n+\frac{c\tau^2}{\sigma^2}\I_n)^{-1}\frac{c\tau^2}{\sigma^2}\I_n = \frac{c\tau^2}{\sigma^2+c\tau^2}\I_n\,,\]
which implies the desired convergence of $H_{ii}$. 
The statement about $AV^{(post)}_i$ follows by combining the above result with Theorem \ref{thm:leverage}.
\end{proof}

\begin{proof}[Proof of Theorem \ref{thm:limsup_fin}]
The statement follows from part (a) of Theorem \ref{thm:limsup_fin_general}, since \eqref{modelcov} is a special case of \eqref{eq:glm_model}.
\end{proof}

\begin{proof}[Proof of Theorem \ref{thm:limsup_fin_general}]
First we prove part (a).
Using $\pi_i^{-1}=\left(\sum_{j=1}^np(y_j|y_{-j})^{-1}\right)p(y_i|y_{-i})$ we can re-write the upper bound in \eqref{eq:av_mix_bound_2} as
\begin{equation}\label{eq:av_mix_bound_3}
AV_i^{(mix)}\leq
\left(\sum_{j=1}^np(y_j|y_{-j})^{-1}\right)
\left(
p(y_i|y_{-i})+
\int p(y_i|\theta)p(\theta|y)d\theta
\right)\,.
\end{equation}
By the subadditivity and submultiplicativity of $\limsup$, and monotonicity of $t\mapsto t^{-1}$ on $(0,\infty)$, it follows 
\begin{equation}\label{eq:av_mix_limsup_bound}
\limsup_{p\to\infty}
AV_i^{(mix)}\leq
\left(\sum_{j=1}^n(\liminf_{p\to\infty}p(y_j|y_{-j}))^{-1}\right)
\left(
\limsup_{p\to\infty}p(y_i|y_{-i})+
\limsup_{p\to\infty}\int p(y_i|\theta)p(\theta|y)d\theta
\right)\,.
\end{equation}
We now prove that all terms on the right-hand side are finite.
We have $p(y_i|y_{-i})=p(y)/p(y_{-i})$ where, by \eqref{eq:glm_model}, 
\begin{align*}
p(y)
\;=\;&
\int \prod_{j=1}^n
g(y_j|\eta_j)
p(\eta)d\eta
\quad\hbox{and}\quad
p(y_{-i})
\;=\;
\int \prod_{j\neq i}
g(y_j|\eta_j)
p(\eta)d\eta
\end{align*}
where $p(\eta)=N(\eta;0,A_p)$ with $A_p=\nu_p^2XX^T$ is the prior distribution on $\eta=(\eta_1,\dots,\eta_n)$ induced by the prior on $(\theta_1,\dots,\theta_p)$ and the linear transformation $\eta=X\theta$.
As shown in the proof of Proposition \ref{prop:high_p}, we have $p^{-1}XX^T\to \tau^2 \I_n$ almost surely element-wise as $p\to\infty$, and thus also
$A_p=\nu_p^2XX^T=p\nu_p^2(p^{-1}XX^T)\to c\tau^2\I_n$, which implies that 
$p(\eta)\to N(\eta;0,c\tau^2\I_n)$ almost surely as $p\to\infty$, where the convergence is point-wise in $\eta\in\R^n$. 
Also, since $A_p\to c\tau^2\I_n$ as $p\to\infty$, we have that, almost surely for large enough $p$, $A_p$ is invertible, its determinant satisfies $|A_p|> (c\tau^2/2)^n$ and $(A_p^{-1}-(2c\tau^2)^{-1}\I_n)$ is positive definite.
These observations imply that, almost surely, for large enough $p$ allow we have
$$
p(\eta)
< (\pi c\tau^2)^{-n/2}\exp\left(-(4c\tau^2)^{-1}\|\eta\|^2
\right)\,,
$$
for every $\eta\in\R^n$.
Combining the above bound with the boundedness of the likelihood,  we can apply the dominated convergence theorem and deduce that
$$
p(y)\to\int_{\R^n} \prod_{j=1}^n
g(y_j|\eta_j)
N(\eta;0,c\tau^2\I_n)d\eta
=
\prod_{j=1}^n\int_{\R}
g(y_j|\eta_j)
N(\eta_j;0,c\tau^2)d\eta_j\in(0,\infty)
$$
almost surely as $p\to\infty$.
Applying the same argument to $p(y_{-i})$ we obtain
\begin{equation}\label{eq:py_finite}
p(y_i|y_{-i})
=
\frac{p(y)}{p(y_{-i})}
\to
\frac{\prod_{j=1}^n\int_{\R}
g(y_j|\eta_j)
N(\eta_j;0,c\tau^2)d\eta_j}{\prod_{j\neq i}\int_{\R}
g(y_j|\eta_j)
N(\eta_j;0,c\tau^2)d\eta_j}
=
\int_{\R}
g(y_i|\eta_i)
N(\eta_i;0,c\tau^2)d\eta_i\in(0,\infty)\,,
\end{equation}
meaning that 
$\limsup_{p\to\infty}p(y_i|y_{-i})<\infty$ and $(\liminf_{p\to\infty}p(y_i|y_{-i}))^{-1}<\infty$ almost surely.
By the same argument we also have $(\liminf_{p\to\infty}p(y_j|y_{-j}))^{-1}<\infty$ for $j=1,\dots,n$. 
Finally, by \eqref{eq:glm_model} and Bayes Theorem, we can write
$$
\int p(y_i|\theta)p(\theta|y)d\theta
=\frac{\int p(y_i|\theta)p(y|\theta)p(\theta)d\theta}{\int p(y|\theta)p(\theta)d\theta}
=
\frac{
\int_{\R^n}
g(y_i|\eta_i)^2 \prod_{j\neq i}g(y_j|\eta_j)p(\eta)d\eta
}{
\int_{\R^n} \prod_{j=1}^n g(y_j|\eta_j) p(\eta)d\eta
}
$$
and applying dominated convergence arguments analogous to above we obtain
$$
\int p(y_i|\theta)p(\theta|y)d\theta
\to
\frac{
\int_\R g(y_i|\eta_i)^2 N(\eta_i;0,c)d\eta_i
}{
\int_\R g(y_i|\eta_i) N(\eta_i;0,c)d\eta_i
}\in(0,\infty)\,,
$$
which implies that $\limsup_{p\to\infty}\int p(y_i|\theta)p(\theta|y)d\theta$.
Combining the above bounds with \eqref{eq:av_mix_limsup_bound} we deduce $\limsup_{p\to\infty}
AV_i^{(mix)}<\infty$.

Consider now part (b) and assume $\int_{\R}g(y_i|\eta_i)^{-1}\exp(-\delta\eta_i^2)d\eta_i<\infty$ for some  $\delta<(2c\tau^2)^{-1}$.
By \eqref{varpost}
\begin{equation*}
 AV^{(post)}_i+1
 =\int\left(\frac{p(\theta|y_{-i})}{p(\theta|y)}\right)^2p(\theta|y)d\theta
 =\frac{p(y_i|y_{-i})}{p(y_{-i})}\int\frac{p(y_{-i}|\theta)^2}{p(y|\theta)}p(\theta)d\theta\,.
\end{equation*}
By \eqref{eq:py_finite} we have $\lim_{p\to\infty}\frac{p(y_i|y_{-i})}{p(y_{-i})}=a$ for some $a\in(0,\infty)$ . Thus
\begin{equation*}
\limsup_{p\to\infty} AV^{(post)}_i+1=
a \limsup_{p\to\infty}\int\frac{p(y_{-i}|\theta)^2}{p(y|\theta)}p(\theta)d\theta\,.
\end{equation*}
By \eqref{eq:glm_model}
$$
\int\frac{p(y_{-i}|\theta)^2}{p(y|\theta)}p(\theta)d\theta
=
\int_{\R^n}
\frac{\prod_{j\neq i}g(y_j|\eta_j)}{g(y_i|\eta_i)}p(\eta)d\eta
\leq
\left(\prod_{j\neq i}\sup_{\eta_j}g(y_j|\eta_j)\right)
\int_{\R}
g(y_i|\eta_i)^{-1}p(\eta_i)d\eta_i
$$
with $p(\eta)=N(\eta;0,A_p)$ as above and $p(\eta_i)=N(\eta_i;0,a^{(i)}_p)$ where $a^{(i)}_p$ is the $i$-th diagonal term of $A_p$.
By $A_p\to c\tau^2\I_n$ almost surely,  we have $a^{(i)}_p\to c\tau^2$ and thus $(2a^{(i)}_p)^{-1}>\delta$ eventually as $p\to\infty$ since $\delta<(2c\tau^2)^{-1}$.
It follows
\begin{align*}
\limsup_{p\to\infty}
\int_{\R}g(y_i|\eta_i)^{-1}p(\eta_i)d\eta_i
=&
(2\pi c\tau^2)^{-1/2}
\limsup_{p\to\infty}
\int_{\R}g(y_i|\eta_i)^{-1}\exp(-(2a^{(i)}_p)^{-1}\eta_i^2)d\eta_i
\\\leq&
(2\pi c\tau^2)^{-1/2}
\limsup_{p\to\infty}
\int_{\R}g(y_i|\eta_i)^{-1}\exp(-\delta\eta_i^2)d\eta_i<\infty
\,.
\end{align*}
Combining the above inequalities we obtain $\limsup_{p\to\infty} AV^{(post)}_i<\infty$ as desired.

Finally, consider part (b) and assume $\int_{\R}g(y_i|\eta_i)^{-1}\exp(-\delta\eta_i^2)d\eta_i=\infty$ for some $\delta>(2c\tau^2)^{-1}$.
In this case, using that $A_p\to c\tau^2\I_n$ as $p\to\infty$ we have 
\begin{align*}
\limsup_{p\to\infty}
\frac{\prod_{j\neq i}g(y_j|\eta_j)}{g(y_i|\eta_i)}p(\eta)d\eta
=&
(2\pi c\tau^2)^{-n/2}
\limsup_{p\to\infty}
\int_{\R^n}
\frac{\prod_{j\neq i}g(y_j|\eta_j)}{g(y_i|\eta_i)}\exp(-\eta^TA_p\eta)d\eta
\\\geq&
(2\pi c\tau^2)^{-n/2}
\int_{\R^n}
\frac{\prod_{j\neq i}g(y_j|\eta_j)}{g(y_i|\eta_i)}\exp(-\delta\|\eta\|^2)d\eta
\\=
(2\pi c\tau^2)^{-n/2}&
\left(\prod_{j\neq i}
\int_{\R}g(y_j|\eta_j)\exp(-\delta\eta_j^2)d\eta_j\right)
\int_{\R}
g(y_i|\eta_i)^{-1}\exp(-\delta\eta_i^2)d\eta_i
=\infty
\,,
\end{align*}
where we used the fact that $\delta\I_n-A_p$ is eventually positive definite as $p\to\infty$ since $\delta>(2c\tau^2)^{-1}$.
\end{proof}

\begin{proof}[Proof of Theorem \ref{thm:av_loo_general}]
\emph{Part (a).}
We first show that $AV_i^{(loo)}$ diverges as $p\to\infty$. 
A derivation analogous to \eqref{avmix}, with the importance distribution $q_{mix}^{(\balpha)}(\theta)$ replaced by the LOO posterior $p(\theta|y_{-i})$ and some simple algebraic simplifications, leads to 
$$ AV_i^{(loo)} = \int \frac{p(\theta|y)^2}{p(\theta|y_{-i})}d\theta - 1 
=\frac{p(y_{-i})}{p(y)^2}\int p(y_i|\theta)p(y|\theta)p(\theta)d\theta - 1\,,
$$
where we also used the conditional independence assumption $p(y|\theta)=\prod_{i=1}^np(y_i|\theta)$.
Combining the above with \eqref{eq:glm_model} we have
\begin{equation}\label{eq:av_loo_expr}
AV_i^{(loo)} +1
=
\frac{\left(\int\prod_{j\neq i}g(y_j|\eta_j)p(\eta_{-i})d\eta_{-i}\right)\left(\int h_i(\eta)p(\eta)d\eta\right)}
{\left(\int \prod_{j=1}^ng(y_j|\eta_j)p(\eta)d\eta\right)^{2}}\,,
\end{equation}
where $p(\eta_{-i})$ and $p(\eta)$ denote the prior distributions of $\eta_{-i}$ and $\eta$ under \eqref{eq:glm_model}, and $h_i(\eta)=g(y_i|\eta_i)\prod_{j=1}^ng(y_j|\eta_j)$.
Since $p(\eta)=N(\eta;0,A_p)$ with $A_p=\nu_p^2XX^T$ and $p(\eta_{-i})=N(\eta_{-i};0,A_p^{(i)})$ with $A_p^{(i)}=\nu_p^2X_{-i}X_{-i}^T$, we can rewrite $AV_i^{(loo)} +1$ as
\begin{equation}\label{eq:terms_p}
 \sqrt{2\pi p\nu_p^2 \frac{|\frac{1}{p}XX^T|}{|\frac{1}{p}X_{-i}X_{-i}^T|}}
\frac{\left(\int\prod_{j\neq i}g(y_j|\eta_j)K_{-i}(\eta_{-i})d\eta_{-i}\right)\left(\int h_i(\eta)K(\eta)d\eta\right)}
{\left(\int \prod_{j=1}^ng(y_j|\eta_j)K(\eta)d\eta\right)^{2}},
\end{equation}
where $K(\eta)=\exp\left(-\eta^T(2\nu_p^2XX^T)^{-1}\eta\right)$ and $K_{-i}(\eta_{-i})=\exp\left(-\eta_{-i}^T(2\nu_p^2X_{-i}X_{-i}^T)^{-1}\eta_{-i}\right)$. 
We now analyze the limiting behaviour of each term in \eqref{eq:terms_p}.
First we have $\lim_{p\to\infty}  {\frac{|\frac{1}{p}XX^T|}{|\frac{1}{p}X_{-i}X_{-i}^T|}}=\tau^2$ since $p^{-1}XX^T\to\tau^2\I_n$ and $p^{-1}X_{-i}X_{-i}^T\to\tau^2\I_{n-1}$ almost surely, as shown in the proof of Proposition \ref{prop:high_p}, and the determinant is a continuous function.
Note that the latter convergences also imply that $XX^T$ and $X_{-i}X_{-i}^T$ are almost surely eventually invertible as $p\to\infty$ so that $K$ and $K_{-i}$ are well defined.
Then $K(\eta)\leq 1$ implies
\begin{align*}
&\int \prod_{j=1}^ng(y_j|\eta_j)K(\eta)d\eta\leq I_y<\infty \,,
\end{align*}
where $I_y = \int \prod_{j=1}^ng(y_j|\eta_j)d\eta=\prod_{j=1}^n \int g(y_j|\eta_j)d\eta_j$ is a positive and finite constant by the  assumptions in part (a).
Also, $K(\eta)=\exp\left(-\frac{1}{2p\nu_p^2}\eta^T(p^{-1}XX^T)^{-1}\eta\right)$ combined with $p^{-1}XX^T\to\tau^2\I_n$ and $p\nu_p^2\to\infty $ implies that  $K(\eta)\to \exp(0)=1$ for every $\eta\in\R^n$ almost surely as $p\to\infty$, and similarly also $K_{-i}(\eta_{-i})\to \exp(0)=1$ for every $\eta_{-i}\in\R^{n-1}$. 
It follows by Fatou's lemma that
\begin{align*}
&\liminf_{p\to\infty}\int\prod_{j\neq i}g(y_j|\eta_j)K_{-i}(\eta_{-i})d\eta\geq I_{y_{-i}} \quad\hbox{ and }\quad
\liminf_{p\to\infty}\int h_i(\eta)K(\eta)d\eta_{-i}\geq I_{\tilde{y}}\,,
\end{align*}
where $I_{\tilde y} = \int h_i(\eta)d\eta$ and $I_{y_{-i}}=\prod_{j\neq i}\int g(y_j|\eta_j)d\eta_j$ are positive and finite constants by the  assumptions in part (a). 
Combining the above results with \eqref{eq:terms_p}, the submultiplicativity of the $\liminf$ and $p\nu_p^2\to\infty $, we get
$$
\liminf_{p\to\infty}AV_i^{(loo)} +1 \geq 
\sqrt{2\pi\tau^2}
\frac{I_{y_{-i}} I_{\tilde y}}{I_y^{2}}
\liminf_{p\to\infty}\sqrt{p\nu_p^2}=\infty\,, 
$$
as desired.

We now prove that also $AV_i^{(mix)}$ diverges as $p\to\infty$ under the assumptions of part (a). 
By \eqref{avmix} and $\frac{p(\theta|y_{-i})}{q_{mix}(\theta)}\leq \pi_i^{-1}$ we have
$$
AV_i^{(mix)}
\geq
\int\frac{p(\theta|y)^2}{q_{mix}(\theta)}d\theta
-2\int\frac{p(\theta|y_{-i})}{q_{mix}(\theta)}p(\theta|y)d\theta
\geq
\int\frac{p(\theta|y)^2}{q_{mix}(\theta)}d\theta
-2\pi_i^{-1}\,,
$$
which implies
$$
\liminf_{p\to\infty}AV_i^{(mix)}
\geq
\liminf_{p\to\infty}
\int\frac{p(\theta|y)^2}{q_{mix}(\theta)}d\theta
-\frac{2}{\liminf_{p\to\infty}\pi_i}\,.
$$
We now prove that $\int\frac{p(\theta|y)^2}{q_{mix}(\theta)}d\theta$ diverges with $p$ and that $\liminf_{p\to\infty}\pi_i>0$ for every $i$, thus deducing $\lim_{p\to\infty}AV_i^{(mix)}=\infty$ from the inequality above. 
First, by \eqref{eq:glm_model} we have
\begin{align*}
&\int\frac{p(\theta|y)^2}{q_{mix}(\theta)}d\theta
=\frac{\sum_{j=1}^np(y_{-j})}{p(y)^2}\int\frac{\prod_{i=1}^ng(y_i|\eta_i)^2}{\sum_{k=1}^n\prod_{i\neq k}g(y_i|\eta_i)}p(\eta)d\eta=\sum_{j=1}^n\frac{p(y_{-j})\int h(\eta)p(\eta)d\eta}{p(y)^2}
\end{align*}
with $h(\eta)=(\sum_{k=1}^ng(y_k|\eta_k)^{-1})^{-1}\prod_{i=1}^ng(y_i|\eta_i)$. 
Then, using $p(y_{-j})=\int\prod_{i\neq j}g(y_i|\eta_i)p(\eta_{-j})d\eta_{-j}$ with $p(\eta_{-j})=(2\pi \nu_p^2|X_{-j}X_{-j}^T|)^{-(n-1)/2}K_{-j}(\eta_{-j})$ as defined above, we have
\begin{align*}
&\int\frac{p(\theta|y)^2}{q_{mix}(\theta)}d\theta
=\sum_{j=1}^n
\sqrt{2\pi p\nu_p^2 \frac{|\frac{1}{p}XX^T|}{|\frac{1}{p}X_{-j}X_{-j}^T|}}
\frac{\left(\int \prod_{k\neq j}g(y_k|\eta_k)K_{-j}(\eta_{-j})d\eta_{-j}\right)\left(\int h(\eta)K(\eta)d\eta\right)}{\left(\int \prod_{j=1}^ng(y_j|\eta_j)K(\eta)d\eta\right)^2}.
\end{align*}
Proceeding as done above for $AV_i^{(loo)} +1 $, exploiting the almost sure point-wise convergences $K(\eta)\to1$ and $K_{-j}(\eta_{-j})\to1$, one can derive
\begin{align*}
&\liminf_{p\to\infty}\int\frac{p(\theta|y)^2}{q_{mix}(\theta)}d\theta
\geq 
\sum_{j=1}^n
\sqrt{2\pi\tau^2}
\frac{I_{y_{-i}}I_{mix}}{I_y^{2}}
\liminf_{p\to\infty}\sqrt{p\nu_p^2}=\infty\,, 
\end{align*}
where $I_{mix}=\int h(\eta)p(\eta)d\eta$ is a positive constant.

We now prove that $\liminf_{p\to\infty}\pi_i>0$ for every $i$. 
From $\pi_i = \frac{p(y_{-i})}{\sum_{j=1}^np(y_{-i})}= (1+\sum_{j\neq i}\frac{p(y_{-j})}{p(y_{-i})})^{-1}$ it follows that
$$
\liminf_{p\to\infty}\pi_i = \left(1+\limsup_{p\to\infty}\sum_{j\neq i}\frac{p(y_{-j})}{p(y_{-i})}\right)^{-1}
\geq
\left(1+\sum_{j\neq i}\limsup_{p\to\infty}\frac{p(y_{-j})}{p(y_{-i})}\right)^{-1}\,.
$$
Then we write  for every $j\neq i$
$$ \frac{p(y_{-j})}{p(y_{-i})} = \left(\frac{|p^{-1}X_{-i}X_{-i}^T|}{|p^{-1}X_{-j}X_{-j}^T|}\right)^{1/2} \frac{\int \prod_{k\neq j}g(y_k|\eta_k)K_{-j}(\eta_{-j})d\eta_{-j}}{\int \prod _{k\neq i}g(y_k|\eta_k)K_{-i}(\eta_{-i})d\eta_{-i}},$$
which, using $p^{-1}X_{-i}X_{-i}^T\to\tau^2\I_{n-1}$, $p^{-1}X_{-j}X_{-j}^T\to\tau^2\I_{n-1}$, $K_{-j}(\eta_{-j})\leq 1$ and $K_{-i}(\eta_{-i})\to 1$, similarly to before, implies that
$\limsup_{p\to\infty} \frac{p(y_{-j})}{p(y_{-i})} 
\leq \frac{I_{y_{-j}}}{I_{y_{-i}}}<\infty$.

\emph{Part (b).}
We start by proving $\limsup_{p\to\infty}AV^{(loo)}_i<\infty$. 
By \eqref{eq:av_loo_expr} we can deduce
\begin{equation}\label{eq:av_loo_Bn}
\limsup_{p\to\infty}AV_i^{(loo)} +1
\leq
B^{2}
\left(\liminf_{p\to\infty}
\int \prod_{j=1}^ng(y_j|\eta_j)p(\eta)d\eta\right)^{-2}
\end{equation}
where $B=\sup_{\eta\in\R^n}\prod_{i=1}^np(y_i|\eta_i)$ is a finite constant by the assumption of upper bounded likelihood. 
Since $p(\eta)=N(\eta;0, \nu_p^2XX^T)\to 0$ almost surely for every $\eta\in\R^n$ as $p\to\infty$, it is convenient to define the change of variables $\gamma = (p\nu_p^2)^{-1/2} \eta$ and re-write the integral above as
\begin{equation}\label{eq:change_of_var}
\int_{\R^n} \prod_{j=1}^ng(y_j|\eta_j)N(\eta;0, \nu_p^2XX^T)d\eta
=
\int_{\R^n}\prod_{j=1}^ng\left(y_j|\sqrt{p\nu_p^2}\gamma_j\right)N(\gamma;0, p^{-1}XX^T)d\gamma.
\end{equation}
Defining $a_i = \lim_{\eta_i\to -\infty}g(y_i|\eta_i)$ and $b_i= \lim_{\eta_i\to\infty}g(y_i|\eta_i)$, we have
$\lim_{p\to\infty}g\left(y_j|\sqrt{p\nu_p^2}\gamma_j\right)=  (a_i (1-\sign(\gamma_i))+b_i\sign(\gamma_i))$
for every $i$ and every $\gamma_i\neq 0\in\R$
and
$\lim_{p\to\infty}N(\gamma;0, p^{-1}XX^T)= N(\gamma;0, \tau^2\I_n)$
for
every $\gamma\in\R^n$
almost surely as $p\to\infty$. Thus, by Fatou's lemma we have
\begin{align}
&\liminf_{p\to\infty} \int_{\R^n}\prod_{j=1}^ng\left(y_j|\sqrt{p\nu_p^2}\gamma_j\right)N(\gamma;0, p^{-1}XX^T)d\gamma \geq\nonumber\\
&\int_{\R^n}
\prod_{j=1}^n (a_i (1-\sign(\gamma_i))+b_i\sign(\gamma_i))
N(\gamma;0, \tau^2\I_n)d\gamma=\prod_{j=1}^n \left( \frac{a_i}{2} + \frac{b_i}{2} \right)  > 0\,.\label{eq:pos_prod_aibi}
\end{align}
The latter product is a positive constant by the assumption $a_i+b_i>0$ for any $i$. 
Combining \eqref{eq:pos_prod_aibi} and \eqref{eq:av_loo_Bn} we obtain $\limsup_{p\to\infty}AV^{(loo)}_i<\infty$ as desired.

We now prove $\limsup_{p\to\infty} AV^{(mix)}_i <\infty$.
Equation \eqref{eq:pos_prod_aibi} states that $\liminf_{p\to\infty} p(y) > 0 $. An analogous derivation can be used to prove that $\liminf_{p\to\infty} p(y_{-j}) > 0 $ for every $j=1,\dots,n$.
Combining the latter with $\limsup_{p\to\infty} p(y_{-j}) \leq B_{-j}<\infty$ for every $j=1,\dots,n$, with $B_{-j}=\sup_{\eta\in\R^n}\prod_{i\neq j}p(y_i|\eta_i)<\infty$, we obtain that $\liminf_{p\to\infty}\pi_i \geq \frac{\liminf_{p\to\infty}p(y_{-i})}{\sum_{j=1}^n\limsup_{p\to\infty}p(y_{-i})}>0$.
One can then deduce
$$\limsup_{p\to\infty} AV^{(mix)}_i < (\limsup_{p\to\infty}\pi_i^{-1} )( \limsup_{p\to\infty} AV^{(loo)}_i)<\infty$$
as desired.

To conclude, we prove $\lim_{p\to\infty} AV^{(post)}_i =\infty$.
By \eqref{eq:glm_model} and \eqref{varpost}
$$ AV_i^{(post)}+1 =  \frac{p(y)}{p(y_{-i})^2} \int_{\R^n} \frac{\prod_{k\neq i}g(y_k|\eta_k)}{g(y_i|\eta_i)} p(\eta)d\eta.$$
Thus
\begin{equation*}
 \liminf_{p\to\infty}AV_i^{(post)}+1 
\geq \frac{\liminf_{p\to\infty}
p(y)}{B_{-i}^2}
\liminf_{p\to\infty}\int_{\R^n} \frac{\prod_{k\neq i}g(y_k|\eta_k)}{g(y_i|\eta_i)} p(\eta)d\eta,
\end{equation*}
where $B_{-i}<\infty$ and $\liminf_{p\to\infty}p(y)>0$ as shown above.
Using the same change of variable of \eqref{eq:change_of_var} and proceeding as in \eqref{eq:pos_prod_aibi} we obtain
\begin{align*}
\liminf_{p\to\infty}
\int_{\R^n} \frac{\prod_{k\neq i}g(y_k|\eta_k)}{g(y_i|\eta_i)} p(\eta)d\eta
\geq 
\left(\frac{1}{2a_i} + \frac{1}{2b_i}\right)\prod_{j\neq i} \left(\frac{a_j}{2}+\frac{b_j}{2}\right)=\infty
\end{align*}
where the latter equality follows from the assumptions that $a_ib_i=0$ and $(a_i+b_i)\in(0,\infty)$.
It follows that $\liminf_{p\to\infty}AV_i^{(post)}=\infty$ almost surely, and thus also $\lim_{p\to\infty}AV_i^{(post)}=\infty$ almost surely as desired.
\end{proof}


\bibliographystyle{apalike}
\bibliography{loo_cv}

\end{document}